\documentclass[11pt]{article}

\usepackage{authblk}

\usepackage{amsmath,amssymb,amsthm}
\usepackage{booktabs}
\usepackage{multirow}
\usepackage{makecell}
\usepackage{enumerate}
\usepackage{graphicx}
\graphicspath{{figs/}}
\usepackage{url}
\usepackage{color,subcaption}
\usepackage{geometry}
\usepackage{algorithm}
\usepackage{algpseudocode}
\usepackage{caption}
\usepackage[counterclockwise]{rotating}
\usepackage[T1]{fontenc}
\usepackage{setspace}

\usepackage{thmtools}
\usepackage{thm-restate}
\usepackage{hyperref}

\usepackage[flushleft]{threeparttable}

\usepackage[round]{natbib}

\hypersetup{pdfborder = {0 0 0},colorlinks=true,linkcolor=blue,citecolor=blue}
\newtheorem{theorem}{Theorem}
\newtheorem{proposition}[theorem]{Proposition}
\newtheorem{lemma}{Lemma}
\newtheorem{corollary}[theorem]{Corollary}
\theoremstyle{definition}

\newtheorem{simulation}{Simulation}

\newcommand\independent{\protect\mathpalette{\protect\independenT}{\perp}}
\def\independenT#1#2{\mathrel{\rlap{$#1#2$}\mkern2mu{#1#2}}}

\title{Empirical Likelihood-Based Estimation and Inference in Randomized Controlled Trials with High-Dimensional Covariates}

\author{Wei Liang and Ying Yan\thanks{Corresponding Author. Email: \href{mailto: yanying7@mail.sysu.edu.cn}{yanying7@mail.sysu.edu.cn}} \\
School of Mathematics, Sun Yat-sen University, No. 135, Xingang Xi Road, Guangzhou 510275, China}

\date{}

\begin{document}

\maketitle

\doublespacing

\begin{abstract}
  In this paper, we propose a data-adaptive empirical likelihood-based approach for treatment effect estimation and inference, which overcomes the obstacle of the traditional empirical likelihood-based approaches in the high-dimensional setting by adopting penalized regression and machine learning methods to model the covariate-outcome relationship. In particular, we show that our procedure successfully recovers the true variance of Zhang's treatment effect estimator \citep{Zhang2018} by utilizing a data-splitting technique. Our proposed estimator is proved to be asymptotically normal and semiparametric efficient under mild regularity conditions. Simulation studies indicate that our estimator is more efficient than the estimator proposed by \citet{Wager2016} when random forests are employed to model the covariate-outcome relationship. Moreover, when multiple machine learning models are imposed, our estimator is at least as efficient as any regular estimator with a single machine learning model. We compare our method to existing ones using the ACTG175 data and the GSE118657 data, and confirm the outstanding performance of our approach.
\end{abstract}

\newpage
\section{Introduction}\label{sec:intro}
Randomized controlled trials (RCTs) are recognized as the standard clinical design to eliminate sources of confounding bias. When the outcome of interest is a continuous variable, the difference of mean responses in the treatment and the control groups is an
unbiased and consistent estimator for the average treatment effect (ATE), a commonly used estimand to evaluate the effect of a treatment or policy. When the baseline information is involved before receiving the treatment, such as age, sex, and other characteristics, adjusting for the pre-treatment covariates helps to improve the efficiency of the ATE estimator.\par
The key of covariate adjustment is to explore the relationship between the auxiliary covariates and response. Analysis of covariance (ANCOVA) is a classical regression method for covariate adjustment where a linear regression model for $\mathrm{E}[Y|X,D]$ is postulated, i.e.,
		\begin{equation}
		\label{ancova}
		\mathrm{E}[Y|X,D]=\beta_0+\beta^\tau_{x}X+\beta_{d}D.
		\end{equation}
Here, $Y$ is the outcome variable, $X$ is the vector of covariates and $D$ is the binary treatment indicator variable. The parameter of interest, the unconditional population-level treatment effect, is $\beta_{d}$. Then, we can make inference about ATE based on the asymptotic normality of the least square estimator, $\widehat{\beta}_{d}^{\text{ols}}$, of $\beta_d$ in (\ref{ancova}) \citep{Imbens2015}. It follows from \citet{Leon2003} and \citet{Tsiatis2008} that $\widehat{\beta}_{d}^{\text{ols}}$ belongs to the class of all regular and asymptotically linear estimators, and more efficient estimators in this class can be obtained by positing two separate working regression models for $\eta^{(1)}(x)=\mathrm{E}[Y|X=x,D=1]$ and $\eta^{(0)}(x)=\mathrm{E}[Y|X=x,D=0]$, respectively.\par
Empirical likelihood (EL) is an alternative way to carry out covariate adjustment. EL was introduced by Owen \citep{Owen1988,Owen1990,Owen2001} and primarily used to construct confidence intervals for the mean or parameters in the general estimating functions \citep{Qin1994}. It has also been adopted as a tool to efficiently incorporate information of auxiliary covariates in causal inference problems \citep{Huang2008,Qin2007}. In particular, when multiple parametric regression models are imposed into constraints, the EL estimator has good performance as long as one of multiple models correctly specifies the covariate-outcome relationship without requiring the knowledge of which model is correct. This is known as the multiple robustness property \citep{Han2013}. Recently, \citet{Zhang2018} and \citet{Tan2020Arxiv} extended EL for statistical inference of ATE in RCTs. Their EL estimators have two prominent advantages. First, \citet{Zhang2018} proved that his EL estimator was at least as efficient as existing regular estimators when the parametric models for the covariate-outcome relationship were mis-specified and asymptotically as efficient as the semiparametric estimator of \citet{Tsiatis2008} when the parametric models for the covariate-outcome relationship were correctly specified. As shown in Tan and Zhang's simulation studies, both Zhang and Tan's EL estimators were considerably more efficient than the estimator of \citet{Tsiatis2008} when the imposed parametric regression models were mis-specified. Secondly, \citet{Tan2020Arxiv} showed that the multiple robustness property of the EL estimator could be maintained in RCTs. \par
In practice, the true covariate-outcome model is unknown, which can be much more complicated than a simple linear combination of several variables in equation (\ref{ancova}). Furthermore, in the big data era, the number of features may be high-dimensional, where ANCOVA and other traditional methods are no longer directly applicable. It inspires us to model the highly complex covariate-outcome relationship by modern machine learning (ML) methods, such as Lasso \citep{Tibshirani1996}, SCAD \citep{Fan2001}, and random forests \citep{Breiman2001}. A general semiparametric framework for statistical inference of treatment effects under which infinite-dimensional nuisance parameters are modelled with ML methods is given by \citet{Chernozhukov2018} and \citet{Belloni2017}, where two crucial points are presented:
		\begin{itemize}
			\item[1.] They used Neyman orthogonal scores to remove the bias brought by regularization.
			\item[2.]They split data to avoid overfitting.
		\end{itemize}
Specifically, the Neyman orthogonal scores technique adjusts for the effect of covariates to reduce sensitivity with respect to the nuisance parameters, and thus promotes the efficiency of treatment effect estimation. It is straightforward to show that the score function developed by \citet{Tsiatis2008} is Neyman orthogonal in RCTs. With an additional data-splitting procedure, \citet{Wager2016} generalized the results of \citet{Tsiatis2008} to the high-dimensional setting and adopted ML methods to model the covariate-outcome relationship. Under mild regularity conditions, they derive valid inference of ATE due to the data-spitting procedure.\par
EL and Neyman orthogonal scores play similar roles in RCTs as they both achieve the goal of efficiency improvement of treatment effect estimation by incorporating information of auxiliary covariates. However, the estimator proposed by Wager and his colleagues does not enjoy some unique properties of EL, e.g., multiple robustness. When the single ML algorithm adopted by \citet{Wager2016} does not successfully capture the covariate-outcome relationship, their ATE estimation may incur efficiency loss. We are motivated to propose a Machine Learning and Data-splitting based Empirical Likelihood (MDEL) approach to estimate ATE, where we apply multiple ML algorithms to model the covariate-outcome relationship. Compared with the high-dimensional regression adjustment approach of \citet{Wager2016}, our proposed EL approach has the following advantages:
		\begin{itemize}
			\item[1.] When the single ML estimator of nuisance parameters does not perform well, our proposed EL estimator is more efficient, as indicated by our simulation studies.
			\item[2.] Different estimators of the nuisance parameters can be imposed simultaneously into constraints to enhance the performance of our estimator. Our simulation studies indicate that our EL estimator with multiple models tends to perform as good as that with the correct model without requiring the knowledge of which model is correct.
		\end{itemize}\par
Our paper is organized as follows. In section \ref{sec:review} we give a brief introduction to our concerning problems and notations. In addition, we review the semiparametric method proposed by \citet{Wager2016}. In section \ref{sec:method}, we introduce our proposed empirical likelihood approach. Then we discuss the practical implementation of our EL approach in section \ref{sec:pract}. In section \ref{sec:rda and sim},  we compare our proposed EL approach to the existing ones in extensive simulation studies, the ACTG175 data set and the GSE118657 data set.

\section{Notations and Reviews}\label{sec:review}
\subsection{The Model Setup}
We introduce our model setup under the potential outcome framework of \citet{Imbens2015}. Suppose we have $n$ observations $\{W_i=(Y_i,X_i,D_i),i=1,\cdots,n\}$ from a binary experiment with the treatment indicator variable $D_i\in\{0,1\}$. $D_i$ takes on the value $1$ if the $i$-th unit is assigned to the treatment group and $0$ if the $i$-th unit is assigned to the control group. Assume $Y_i(d)$ is the potential outcome under $D_i=d$ for $d=0,1$. The observed outcome of the $i$-th unit, $Y_i$, satisfies $Y_i=D_iY_i(1)+(1-D_i)Y_i(0)$. $X_i$ is the covariates of the $i$-th unit and contains additional pre-treatment information. We are interested in the estimation and inference of the population level ATE, defined by $\theta=\mathrm{E}[Y_i(1)-Y_i(0)]$, under the assumption that $\{W_i,i=1,\cdots,n\}$ are independent and identically distributed random samples from $W=(Y,X,D)$. In this paper, we focus on randomized controlled trials, where $D_i$ is randomly assigned to either 0 or 1 and is independent of all pre-treatment variables and the potential outcomes, i.e.,
		\[
		D_i\independent\left\{Y_i(1),Y_i(0), X_i\right\}\quad\text{for}\quad i=1,\cdots,n.
		\]
Let $\delta=\mathrm{P}(D=1)$ be the probability of a unit being assigned to the treatment group. We further assume that $0<\delta<1$. The Radon-Nikodym Theorem indicates that $\theta=\mathrm{E}[Y|D=1]-\mathrm{E}[Y|D=0]$, which leads to a natural and commonly used consistent estimator of ATE, the difference in the means, defined by
		\[
		\widehat{\theta}_{\text{dim}}=\bar{Y}^{(1)}-\bar{Y}^{(0)}=\sum\limits_{i=1}^n\frac{D_iY_i}{n_1}-\sum\limits_{i=1}^n\frac{(1-D_i)Y_i}{n_0},
		\]
where $n_1$ is the size of the treatment group and $n_0=n-n_1$.
$\widehat{\theta}_{\text{dim}}$ ignores information of covariates and thus loses efficiency. To exploit information of covariates and regain the efficiency, we review a popular regression adjustment approach in the following section.

\subsection{A Regression Adjustment Approach}
There are many useful regression adjustment methods when the covariates are low-dimensional. Here, we only review the method of \citet{Tsiatis2008}, which can be reformulated as a method based on the efficient score. For other methods and more details, we refer to \citet{Zhang2018} and \citet{Tan2020Arxiv}. Different from ANCOVA, the approach of \citet{Tsiatis2008} focuses on separately modelling the covariate-outcome relationships $\eta^{(d)}(x)=\mathrm{E}[Y|D=d,X=x],d=0,1$. By fitting $\eta^{(1)}(x)$ and $\eta^{(0)}(x)$ with two different parametric models ${f}_{1}(x,{\alpha}_1)$ and ${f}_{0}(x,{\alpha}_0)$, specified by two finite dimensional parameters, $\alpha_1$ and $\alpha_0$, respectively, \citet{Tsiatis2008} proposed to estimate $\theta$ with
		\[
		\widehat{\theta}_{\text{tdzl}}=\bar{Y}^{(1)}-\bar{Y}^{(0)}-\sum_{i=1}^{n}\left(D_{i}-\frac{n_1}{n}\right)\left(n_{0}^{-1} {f}_{0}(X_i,\widehat{\alpha}_0)+n_{1}^{-1} {f}_{1}(X_i,\widehat{\alpha}_1)\right),
		\] 
where $\widehat{\alpha}_d$, $d=0,1$, are estimators of $\alpha_0$ and $\alpha_1$, respectively, e.g., the least square estimator or the stepwise regression estimator.  Write $\widehat{f}_d(\cdot)=f(\cdot,\widehat{\alpha}_d),d=0,1$. The semiparametric theory \citep{Tsiatis2007} indicates that the efficient score of $\theta$ is given by
		\[
		\varphi(W,\theta,\delta,\eta^{(1)},\eta^{(0)})=\frac{D}{\delta}\left(Y-\eta^{(1)}(X) \right)-\frac{1-D}{1-\delta}\left(Y-\eta^{(0)}(X) \right) + \eta^{(1)}(X) - \eta^{(0)}(X) - \theta.
		\]
Here, $\eta^{(1)}$ and $\eta^{(2)}$ are treated as nuisance parameters, and $\theta$ is the parameter of interest. $\widehat{\theta}_{\text{tdzl}}$ can be reformulated as the solution of
		\[
		\frac{1}{n}\sum_{i=1}^{n}\varphi(W_i,\theta,\widehat{\delta}=\frac{n_1}{n},\widehat{f}_{1},\widehat{f}_{0}) =0.
		\]
Therefore, $\widehat{\theta}_{\text{tdzl}}$ reaches the semiparametric efficiency bound if both $f_1$ and $f_0$ are correctly specified.
\subsection{High-Dimensional Regression Adjustment}
\citet{Wager2016} extended the approach of \citet{Tsiatis2008} to the high-dimensional case. The nuisance parameters were proposed to be estimated using ML methods and a data-splitting procedure was adopted for valid inference with high-dimensional covariates. Let $\mathbb{I}=\{1,\cdots,n\}$ be the sample index set, $\mathbb{I}^{(1)}$ and $\mathbb{I}^{(0)}$ be the index set of the treatment group and control group, respectively. We use the notation $|\mathbb{A}|$ as the size of a set $\mathbb{A}$. Suppose we randomly partition $\mathbb{I}^{(d)}$ into $K$ subsets with equal size, denoted by $(\mathbb{I}^{(d)}_k)_{k=1}^K$, for $d=0,1$. Let $\mathbb{I}_k=\mathbb{I}^{(1)}_{k}\cup \mathbb{I}^{(0)}_{k}$, ${\mathbb{I}^{(d)}_{k}}^c=\mathbb{I}^{(d)}\backslash \mathbb{I}^{(d)}_{k}$ and $\mathbb{I}_{k}^c=\mathbb{I}\backslash \mathbb{I}_{k}$. Generally, we set $\frac{|\mathbb{I}^{(d)}_k|}{|\mathbb{I}_k|}=\frac{n_d}{n}$ and $\frac{|\mathbb{I}^{(d)}_k|}{|\mathbb{I}^{(d)}|}=\frac{1}{K}$.\par
After data-splitting, \citet{Wager2016} proposed to estimate $\theta$ with $\widehat{\theta}_{\text{wdtt}}=\frac{1}{K}\sum\limits_{k=1}^K\widehat{\theta}_{\text{wdtt}}^k$, where the $k$-th sub-estimator, $\widehat{\theta}_{\text{wdtt}}^k$, is the solution of
		\[
	   	\frac{1}{|\mathbb{I}_k|}\sum_{i\in\mathbb{ I}_k}\varphi(W_i,{\theta},\widehat{\delta}=\frac{|\mathbb{I}^{(1)}_k|}{|\mathbb{I}_k|},\widehat{\textsl{g}}^{(1)}_{k},\widehat{\textsl{g}}^{(0)}_{k}) =0.
		\]
For fixed $k$ and $d$, $\widehat{\textsl{g}}^{(d)}_{k}$ is an ML estimator of $\eta^{(d)}$ obtained via the sub-sample $\left(W_i\right)_{{i\in \mathbb{I}^{(d)}_{k}}^c}=\left\{W_i|i\in {\mathbb{I}^{(d)}_{k}}^c \right\}$. It follows immediately that conditional on the sample $\left(W_i\right)_{{i\in \mathbb{I}^{(d)}_{k}}^c}$, $\widehat{\textsl{g}}^{(d)}_{k}(x)$ is a non-random function of $x$. Therefore, the variance of $\widehat{\theta}_{\text{wdtt}}$ can be directly estimated by
		\[
		\widehat{\mathrm{Var}}(\widehat{\theta}_{\text{wdtt}})=\sum\limits_{k=1}^K\frac{|\mathbb{I}_k|^2}{n^2}\widehat{\mathrm{Var}}(\widehat{\theta}_{\text{wdtt}}^k)
		\]
where for a fixed $k$, $\widehat{\mathrm{Var}}(\widehat{\theta}_{\text{wdtt}}^k)$ is a moment-based plug-in variance estimator for the conditional variance of $\widehat{\theta}_{\text{wdtt}}^k$, 
        \[
        \sum\limits_{d=0,1}\frac{1}{|\mathbb{I}^{(d)}_k|}{\mathrm{Var}}\left[\left.Y-\frac{|\mathbb{I}^{(0)}_k|}{|\mathbb{I}_k|}\widehat{\textsl{g}}^{(1)}_k(X)-\frac{|\mathbb{I}^{(1)}_k|}{|\mathbb{I}_k|}\widehat{\textsl{g}}^{(0)}_k(X)\right\vert \widehat{\textsl{g}}^{(1)}_k,\widehat{\textsl{g}}^{(0)}_k,D=d\right].
        \]
\citet{Wager2016} demonstrated that $\frac{\left(\widehat{\theta}_{\text{wdtt}}-\theta\right)}{\sqrt{{\widehat{\mathrm{Var}}(\widehat{\theta}_{\text{wdtt}})}}}$ was asymptotically standard normal under certain regularity conditions. Therefore, for statistical inference, the corresponding $1-\alpha$ confidence interval for $\theta$ is given by
		\[
		\left(\widehat{\theta}_{\text{wdtt}}-z_{\frac{\alpha}{2}}\sqrt{\widehat{\mathrm{Var}}(\widehat{\theta}_{\text{wdtt}})},\widehat{\theta}_{\text{wdtt}}+z_{\frac{\alpha}{2}}\sqrt{\widehat{\mathrm{Var}}(\widehat{\theta}_{\text{wdtt}})}\right),
		\]
where $z_{\frac{\alpha}{2}}$ is the upper quantile of the standard normal distribution.

\section{Empirical Likelihood-Based Approaches in RCTs}\label{sec:method}
\subsection{Traditional EL-Based Approaches in RCTs}\label{subsec:tea}
Let $f(x)=(f_1(x),f_0(x))^\tau$ be a vector function of $x$ and $\xi=\mathrm{E}[f(X)]$. Based on two unbiased estimating functions
		\[
        h_1(D,Y,\theta,\delta)=\frac{DY}{\delta}-\frac{(1-D)Y}{1-\delta}-\theta\quad\text{and}\quad h_2(D,X,\delta,f,\xi)=\frac{D-\delta}{\delta(1-\delta)}\{f(X)-\xi\},
		\]
\citet{Zhang2018} proposed to estimate $\theta$ by maxmizing the nonparametric likelihood $L_F=\prod_{i=1}^{n}p_i$ subject to constraints $\sum_{i=1}^{n}p_i=1,p_i\geq 0$ and $\sum_{i=1}^{n}p_i\left(h_1(D_i,Y_i,\theta,\hat{\delta}),h^\tau_2(D_i,X_i,\hat{\delta},\hat{f},\hat{\xi})\right)^\tau=0$, where $\hat{\delta}=\frac{n_1}{n}$, $\hat{\xi}=\frac{1}{n}\sum_{i=1}^{n}\hat{f}(X_i)$, $\hat{f}=(\hat{f}_1,\hat{f}_0)^\tau$, and $\hat{f}_1$ and $\hat{f}_0$ are estimated working regression models for $\eta^{(1)}$ and $\eta^{(0)}$, respectively. The variance of $\widehat{\theta}_{\text{Zhang}}$ was estimated by a sandwich variance estimator. When the covariates are of low-dimension and $\hat{f}_d$ is an estimated parametric regression model for $\eta^{(d)}$,  $\widehat{\theta}_{\text{Zhang}}$ is more efficient than the semiparametric estimator $\widehat{\theta}_{\text{tdzl}}$ as suggested by Zhang's simulation studies. However, when the covariates are of high-dimension and $\hat{f}_d$ is an ML estimator for model selection, many spurious variables which have high correlations with the response but do not belong to the true feature set will be selected and thus result in serious underestimation of the variance \citep{Fan2012}. We conduct simulations to illustrate this point in Section \ref{subsec: vr}.\par
\citet{Tan2020Arxiv} extended the two-sample EL approach of \citet{Wu2012} and proposed to estimate $\theta$ based on the property
		\[
	    \mathrm{E}\left[f(X)-\xi|D=d\right]=0.
		\]
Their estimator is $\widehat{\theta}_{\text{Tan}}=\sum_{i\in\mathbb{I}^{(1)}}\hat{p}_iY_i-\sum_{i\in\mathbb{I}^{(0)}}\hat{p}_iY_i$, where $\hat{p}_i,i\in\mathbb{I}^{(d)}$ are obtained by maxmizing the nonparametric likelihood $L_F=\prod_{i\in\mathbb{I}^{(d)}}p_i$ subject to constraints 
		$\sum_{i\in \mathbb{I}^{(d)}}p_i=1,p_i\geq 0$ and
	    $
	    \sum_{i\in\mathbb{I}^{(d)}}p_i\hat{f}_d(X_i)=\frac{1}{n}\sum_{j=1}^{n}\hat{f}_d(X_j).
	    $
Here $\hat{f}_d(x)$ is a guess of $\mathrm{E}[Y|X=x,D=d]$. Multiple guesses are allowed in Tan's method. Estimation for the variance of $\widehat{\theta}_{\text{Tan}}$ is given by the bootstrap method. Tan's approach is simple and easy to explain. Asymptotic theory and simulation studies of \citet{Tan2020Arxiv} verify its multiple robustness property, which means that the estimator achieves the semiparametric efficiency bound as long as one model of $f_d$ is correctly specified. However, when $\hat{f}_d$ involves ML estimators, their proposed bootstrap re-sampling procedure is no longer applicable as Donsker conditions are inappropriate when the space of $\hat{f}_d$ is highly complicated.\par
As we can see, both EL approaches in RCTs have desirable properties in the low-dimensional setting but fail to make valid inference in the high-dimensional setting. To maintain multiple robustness and other ideal properties of EL estimators, as well as to overcome the invalid inference problem of traditional EL approaches, we are motivated to extend the approach of \citet{Tan2020Arxiv}, which is very simple to implement, to RCTs with high-dimensional covariates by means of machine learning and data-splitting.

\subsection{The Proposed EL-Based Estimator with High-Dimensional Covariates}\label{subsec:aea}
In our proposed approach, the nuisance parameters are allowed to be estimated using multiple ML methods. For $d=0,1$, assume we already have an $r$-dimensional vector of estimators of $\eta^{(d)}$, denoted as $\widehat{\textsl{g}}^{(d)}_k=\left(\widehat{\textsl{g}}^{(d)}_{k,1},\cdots,\widehat{\textsl{g}}^{(d)}_{k,r}\right)^\tau$, where each component of $\widehat{\textsl{g}}^{(d)}_{k}$ is an ML estimator such as the random forests estimator or Lasso estimator of ${\eta}^{(d)}$ based on the sub-sample $\left(W_i\right)_{{i\in \mathbb{I}^{(d)}_{k}}^c}$. Let $\widehat{\xi}^{(d)}=\frac{1}{n}\sum\limits_{k=1}^K\sum\limits_{i\in \mathbb{I}_k}\widehat{\textsl{g}}^{(d)}_{k}(X_i)$ and
$\widehat{\xi}^{(d)}_k=\frac{1}{|\mathbb{I}_k|}\sum\limits_{i\in \mathbb{I}_k}\widehat{\textsl{g}}^{(d)}_{k}(X_i)$ for $k=1,\cdots,K$ and $d=0,1$. It is easy to check that $\widehat{\xi}^{(d)}=\frac{1}{K}\sum\limits_{k=1}^K\widehat{\xi}^{(d)}_k$. Let $k_d(y,x)$ be the conditional density of $(Y,X)$ given $D=d$ and $p_i=k_d(Y_i,X_i)$, $i\in\mathbb{I}^{(d)}$ for $d=0,1$. Due to randomization, it is easy to verify that, conditional on the sub-sample $\left(W_i\right)_{{i\in \mathbb{I}^{(d)}_{k}}^c}$,
	     \[
	     \mathrm{E}\left[\left.\widehat{\textsl{g}}^{(d)}_{k}(X)\right\vert D=d\right]=\mathrm{E}\left[\widehat{\textsl{g}}^{(d)}_{k}(X)\right]
	     \]
for $k=1,\cdots,K$ and $d=0,1$. Then, the empirical form of the above equation is naturally given by $\sum\limits_{i\in\mathbb{I}_k^{(d)}}p_{i}\widehat{\textsl{g}}^{(d)}_{k}(X_i)=\frac{1}{K}\widehat{\xi}_k^{(d)}$ for $k=1,\cdots,K$ and $d=0,1$, which leads to, simply by summation,
	     \begin{equation}\label{eq:prandom}
	      \sum\limits_{k=1}^K\sum\limits_{i\in \mathbb{I}^{(d)}_k}p_{i}\widehat{\textsl{g}}^{(d)}_{k}(X_i)=\frac{1}{K}\sum_{k=1}^{K}\widehat{\xi}_k^{(d)}  
	     \end{equation}
for $d=0,1$. Based on (\ref{eq:prandom}) and $\iint k_d(y,x)dydx=1$, we impose the following constraints on $p_i$:
	    \begin{equation}
	    \begin{aligned}
		\quad &\sum\limits_{i\in \mathbb{I}^{(d)}} p_{i}=1,\quad d=0,1,\quad p_{i}\geq 0,\quad\forall i\in \mathbb{I},\\
		&\sum\limits_{k=1}^K\sum\limits_{i\in \mathbb{I}^{(d)}_k}p_{i}\widehat{\textsl{g}}^{(d)}_{k}(X_i)=\frac{1}{K}\sum_{k=1}^{K}\widehat{\xi}_k^{(d)},\quad d=0,1.
		\end{aligned}
		\label{eq:oriel}		
		\end{equation}
And we propose to estimate $p_i$ by maximizing the conditional nonparametric likelihood $L_k=\prod\limits_{i\in\mathbb{I}^{(1)}}p_i\prod\limits_{i\in\mathbb{I}^{(0)}}p_i$ subject to the constraints (\ref{eq:oriel}), which is equivalent to solving two separated minimization problems:
		\begin{equation}\label{eq:min}
		\begin{aligned}
		\min\limits_{p_i} &-\sum\limits_{i\in \mathbb{I}^{(d)}}\log (p_{i})\\
		\text{s.t.}\quad &\sum\limits_{i\in \mathbb{I}^{(d)}} p_{i}=1,\quad p_i\geq 0,i\in\mathbb{I}^{(d)},\\
		&\sum\limits_{k=1}^K\sum\limits_{i\in \mathbb{I}^{(d)}_k}p_{i}\left(\widehat{\textsl{g}}^{(d)}_{k}(X_i)-\widehat{\xi}^{(d)}\right)=0.
		\end{aligned}		
		\end{equation}
for $d=0,1$.
Let $\widehat{\textsc{G}}\left(x,\widehat{\textsl{g}}^{(d)}_{k},\widehat{\xi}^{(d)}\right)=\widehat{\textsl{g}}^{(d)}_{k}(x)-\widehat{\xi}^{(d)}$. The Lagrange multiplier method shows that the dual problem of (\ref{eq:min}) is 
		\begin{align}\label{eq:dual}
		\max_{\lambda_d}\ell(\lambda_d):\quad\ell(\lambda_d)=-\sum\limits_{k=1}^K\sum_{i\in \mathbb{I}_k^{(d)}} \log \left\{1+\lambda_d^{\tau}\widehat{\textsc{G}}\left(X_i,\widehat{\textsl{g}}^{(d)}_{k},\widehat{\xi}^{(d)}\right)\right\}-n_d\log n_d
		\end{align} 
and $\hat{p}_i$ is given by
		\[
		\widehat{p}_i=\{n_d\left(1+\widehat{\lambda}_d^{\tau} \widehat{\textsc{G}}\left(X_i,\widehat{\textsl{g}}^{(d)}_{k},\widehat{\xi}^{(d)}\right)\right)\}^{-1}\quad\text{for}\quad i\in \mathbb{I}^{(d)}_k,\quad k=1,\cdots,K,
		\]
where $\widehat{\lambda}_d$ is the solution of (\ref{eq:dual}). Simple calculation reveals that $\widehat{\lambda}_d$ is determined by
		\begin{align}\label{eq:lambdas}
		\frac{1}{n_d}\sum\limits_{k=1}^K\sum\limits_{i\in \mathbb{I}_k^{(d)}}\frac{\widehat{\textsc{G}}\left(X_i,\widehat{\textsl{g}}^{(d)}_{k},\widehat{\xi}^{(d)}\right)}{1+\widehat{\lambda}_d^{\tau} \widehat{\textsc{G}}\left(X_i,\widehat{\textsl{g}}^{(d)}_{k},\widehat{\xi}^{(d)}\right)}=0.
		\end{align}
Our proposed Machine learning and Data-splitting based Empirical Likelihood (MDEL) estimator for $\theta$ is $$\widehat{\theta}_{\text{mdel}}=\widehat{\theta}^{(1)}_{\text{mdel}}-\widehat{\theta}^{(0)}_{\text{mdel}}=\sum\limits_{i=1}^nD_i\widehat{p}_iY_i-\sum\limits_{i=1}^n(1-D_i)\widehat{p}_iY_i.$$ To solve the aforementioned optimization problem and obtain the MDEL estimator, we carry out a modified Newton-Raphson algorithm with details extensively discussed in \citet{Wu2004}. The unique and global minimizer of (\ref{eq:min}) requires that $0$ must be contained in the convex hull of\\ $\left\{\widehat{\textsc{G}}\left(X_i,\widehat{\textsl{g}}^{(d)}_{k},\widehat{\xi}^{(d)}\right),i\in\mathbb{I}^{(d)}\right\}$, which is asymptotically guaranteed by the following regularity condition:
\begin{itemize}
    \item[(\textbf{A1})] With probability tending to 1 as $n\rightarrow\infty$, $\mathrm{Var}\left\{\left[\widehat{\textsl{g}}^{(d)}_{k}(X)\right]_j\right\}>0$ and $\mathrm{E}\left[\left[\widehat{\textsl{g}}^{(d)}_k(X)\right]^2_j\right]<\infty$ conditionally on $\left(W_i\right)_{{i\in \mathbb{I}^{(d)}_{k}}^c}$ for $j=1,\cdots,r$, $k=1,\cdots,K$ and $d=0,1$, where $[x]_i$ to denote the i-th element of a vector $x$.
\end{itemize}
Condition (\textbf{A1}) is mild and does not give any restriction on the dimension $p$ of covariates. Therefore, the solution of (\ref{eq:lambdas}) exists and is unique as long as (\textbf{A1}) is satisfied and $n$ is large enough, even if the dimension of covariates grows with the sample size, which is indicated by the following proposition.
\begin{proposition}\label{pro:1}
Under condition (\textbf{A1}), the vector $0$ is in the convex hull of\\ $\left\{\widehat{\textsc{G}}\left(X_i,\widehat{\textsl{g}}^{(d)}_{k},\widehat{\xi}^{(d)}\right),i\in\mathbb{I}^{(d)}\right\}$ with probability tending to 1 as $n\rightarrow\infty$ for $d=0,1$.
\end{proposition}
For $d=0,1$, define 
	 \begin{align*}
	 &\ddot{\xi}^{(d)}=\frac{1}{K}\sum\limits_{k=1}^K\mathrm{E}\left[\left.\widehat{\textsl{g}}^{(d)}_k(X)\right\vert(W_i)_{i\in \mathbb{I}_k^{(d)^c}}\right],\quad \widehat{\textsc{G}}\left(x,\widehat{\textsl{g}}^{(d)}_{k},\ddot{\xi}^{(d)}\right)=\widehat{\textsl{g}}^{(d)}_k(x)-\ddot{\xi}^{(d)},\\
	 &\widehat{V}_n^{(d)}=\frac{1}{n}\sum\limits_{k=1}^K\sum\limits_{i\in \mathbb{I}^{(d)}_k}\frac{\widehat{\textsc{G}}\left(X_i,\widehat{\textsl{g}}^{(d)}_{k},\widehat{\xi}^{(d)}\right)^{\otimes2}}{(2d-1)\delta+1-d},\quad
	 \ddot{J}_n^{(d)}=\frac{1}{n}\sum\limits_{k=1}^K\sum\limits_{i\in \mathbb{I}^{(d)}_k}\frac{(Y_i-\theta_d)\widehat{\textsc{G}}\left(X_i,\widehat{\textsl{g}}^{(d)}_{k},\ddot{\xi}^{(d)}\right)}{(2d-1)\delta+1-d},\\
	 &\ddot{{S}}^{(d)}_n=\frac{1}{n}\sum\limits_{k=1}^K\sum\limits_{i\in \mathbb{I}^{(d)}_k}\frac{\widehat{\textsc{G}}\left(X_i,\widehat{\textsl{g}}^{(d)}_{k},\ddot{\xi}^{(d)}\right)^{\otimes2}}{(2d-1)\delta+1-d},  
	 \end{align*}
	 where, for any vector or matrix ${H}$, ${H}^{\otimes2}={H}{H}^\tau$. To derive the asymptotic properties of the proposed MDEL estimator, we impose the following regularity condition:
	 \begin{itemize}
	     \item[(\textbf{A2})] With probability tending to 1 as $n\rightarrow\infty$, $V_1<\widehat{V}_n^{(d)}<V_2$ where $V_1$ and $V_2$ are two finite positive semi-definite matrices with rank $q_2\geq q_1>0$.
	 \end{itemize}
	 
	 (\textbf{A2}) is also mild and reasonable as we do not require that $\widehat{V}_n^{(d)}$ converges to any finite and positive semi-definite matrix, but is bounded by two finite and positive semi-definite matrices with probability tending to 1. (\textbf{A2}) is similar to the finite covariance matrix of rank $q>0$ condition in the traditional empirical likelihood theory \citep{Owen2001}. (\textbf{A2}) together with (\textbf{A1}) are sufficient for us to derive the order of $\widehat{\lambda}_d$, which is given by the following proposition:
	 \begin{proposition}\label{pro:2}
	 Under conditions (\textbf{A1})-(\textbf{A2}), we have $||\widehat{\lambda}_d||=O_p(\frac{1}{\sqrt{n}})$ for $d=0,1$, where $||\cdot||$ is the Euclidean norm.
	 \end{proposition}
	 To derive the asymptotic decomposition and study the asymptotic properties of the MDEL estimator, the following regularity conditions are needed.
	 \begin{itemize}
     \item[(\textbf{A3})] $\ddot{{S}}^{(d)}_n,d=0,1$ are invertible given a large $n$.
	 \item[(\textbf{A4})] $\mathrm{E}[Y^2|D=d]<\infty.$ for $d=0,1$. 
	 \end{itemize}
	 (\textbf{A3}) is very likely to hold under condition (\textbf{A2}) when $V_1$ and $V_2$ are both finite and positive definite, but there are still some gaps. For simplicity and avoiding lengthy discussions in this paper, we impose (\textbf{A3}) and ignore the details of the connections. (\textbf{A4}) is very basic and common in the empirical likelihood theory.
	 \begin{proposition}\label{pro:3}
	 Under conditions (\textbf{A1})-(\textbf{A4}), given a large $n$, we have
	 \begin{equation*}
     	\begin{aligned}
     		\sqrt{n}\left(\widehat{\theta}_{\text{mdel}}-\theta\right)=\frac{1}{\sqrt{n}}\sum\limits_{k=1}^K\sum\limits_{i\in \mathbb{I}_k}\left[\frac{D_i}{\delta}\left(Y_i-\theta_1\right)-\frac{D_i-\delta}{\delta}\ddot{J}_n^{(1)^\tau}\ddot{{S}}_n^{(1)^{-1}}\widehat{\textsc{G}}\left(X_i,\widehat{\textsl{g}}^{(1)}_{k},\ddot{\xi}^{(1)}\right)\right.\\
     		\left.-\frac{1-D_i}{1-\delta}\left(Y_i-\theta_0\right)+\frac{D_i-\delta}{1-\delta}\ddot{J}_n^{(0)^\tau}\ddot{{S}}_n^{(0)^{-1}}\widehat{\textsc{G}}\left(X_i,\widehat{\textsl{g}}^{(0)}_{k},\ddot{\xi}^{(0)}\right)\right]+o_p(1),
     	\end{aligned}
     	\end{equation*}
     where $\theta_1=\mathrm{E}[Y(1)]$ and $\theta_0=\mathrm{E}[Y(0)]$.
	 \end{proposition}

In the following theorem, we show that our proposed estimator with a single covariate-outcome model is asymptotically normal and semiparametric efficient when the nuisance estimators are consistent.
		\begin{theorem}
		Under regularity conditions (\textbf{A1})-(\textbf{A4}), if $r=1$ and $$\mathrm{E}\left[\left.\left(\widehat{\textsl{g}}^{(d)}_k(X)-\eta^{(d)}(X)\right)^2\right\vert(W_i)_{i\in \mathbb{I}_k^{(d)^c}}\right] \rightarrow 0$$ in probability as $n\rightarrow\infty$ for $k=1,\cdots,K$ and $d=0,1$, $\sqrt{n}\left(\widehat{\theta}_{\text{mdel}}-\theta\right)$ is asymptotically standard normal with mean zero and variance $\mathrm{Var}[\varphi(W,\theta,\delta,\eta^{(1)},\eta^{(0)})]$, where
		\[
		\varphi(W,\theta,\delta,\eta^{(1)},\eta^{(0)})=\frac{D}{\delta}\left(Y-\eta^{(1)}(X) \right)-\frac{1-D}{1-\delta}\left(Y-\eta^{(0)}(X) \right) + \eta^{(1)}(X) - \eta^{(0)}(X) - \theta.
		\]
        Therefore, $\widehat{\theta}_{\text{mdel}}$ achieves the semiparametric efficiency bound.
		\label{thm 1}
	    \end{theorem}
Note that the condition, $\mathrm{E}\left[\left.\left(\widehat{\textsl{g}}^{(d)}_k(X)-\eta^{(d)}(X)\right)^2\right\vert(W_i)_{i\in \mathbb{I}_k^{(d)^c}}\right] \rightarrow 0$ in probability as $n\rightarrow\infty$, called ``risk consistency'' in \citet{Wager2016}, is mild for the penalized regression methods when the regression model is correctly specified and sufficient sparsity is satisfied. And the ``risk consistency'' is also mild for many ML methods such as the random forests and neural networks when sufficient sparsity is satisfied. For more details, we refer to \citet{Chernozhukov2018}. When $r>1$, i.e., multiple models are imposed to estimate nuisance parameters, we expect that Theorem \ref{thm 1} still holds when any one of estimators for nuisance parameters satisfies ``risk consistency'' condition. Moreover, we expect that the convergence rate of the estimator with multiple models is identical to that with the oracle model. However, the asymptotic theory requires further complicated assumptions about the structure of covariates such that the weak law of large numbers can be applied to dependent terms, and we remain it as future work. Instead, we use simulation studies in Section \ref{sec:rda and sim} to show that our proposed estimator attains multiple robustness property and is approximately normally distributed with reasonable coverage rates.

\subsection{Variance Recovery for Valid Inference}\label{subsec: vr}
Based on the decomposition of $\widehat{\theta}_{\text{sel}}$ in Proposition \ref{pro:3}, we propose to estimate the variance of $\widehat{\theta}_{\text{mdel}}$ with
		 \begin{equation}
		 \begin{aligned}
		 \widehat{\sigma}^2_{\text{mdel}}=\frac{1}{n}\sum\limits_{d=0,1}\sum\limits_{k=1}^K\sum\limits_{i\in \mathbb{I}^{(d)}_k}\frac{n_d}{n}\widehat{p}_i\left\{\frac{n}{n_1}D_i(Y_i-\widehat{\theta}_{\text{mdel}}^{(1)})-\frac{n}{n_1}(D_i-\frac{n_1}{n}){\widehat{J}_n^{(1)^\tau}}{\widehat{{S}}_n^{(1)^{-1}}}\widehat{\textsc{G}}\left(X_i,\widehat{\textsl{g}}^{(1)}_{k},\widehat{\xi}^{(1)}\right)-\right.\\
		 \left.\frac{n}{n_0}(1-D_i)(Y_i-\widehat{\theta}_{\text{mdel}}^{(0)})+\frac{n}{n_0}(D_i-\frac{n_1}{n}){\widehat{J}_n^{(0)^\tau}}{\widehat{{S}}_n^{(0)^{-1}}}\widehat{\textsc{G}}\left(X_i,\widehat{\textsl{g}}^{(0)}_{k},\widehat{\xi}^{(0)}\right)\right\}^2,
		 \end{aligned}
		 \end{equation}
where 
		 \begin{align*}
		 \widehat{J}_n^{(d)} = \sum\limits_{k=1}^K\sum\limits_{i\in \mathbb{I}^{(d)}_k}\widehat{p}_iY_i\widehat{\textsc{G}}\left(X_i,\widehat{\textsl{g}}^{(d)}_{k},\widehat{\xi}^{(d)}\right),\\
		 \widehat{{S}}_n^{(d)}=\sum\limits_{v\in\{0,1\}}\sum\limits_{k=1}^K\sum\limits_{i\in \mathbb{I}_k^{(v)}}\frac{n_v}{n}\widehat{p}_i\widehat{\textsc{G}}\left(X_i,\widehat{\textsl{g}}^{(d)}_{k},\widehat{\xi}^{(d)}\right)\widehat{\textsc{G}}\left(X_i,\widehat{\textsl{g}}^{(d)}_{k},\widehat{\xi}^{(d)}\right)^\tau,
		 \end{align*}
for $d=0,1$. In the following theorem, we prove that the variance estimator of our EL approach converges to the true variance asymptotically. That is, our variance estimator successfully recovers the true variance. 
		 \begin{theorem}
Under the assumptions and regularity conditions of Theorem \ref{thm 1}, we have 
		 	\[
		 	{\widehat{\sigma}^2_{\text{mdel}}}\rightarrow\frac{ \mathrm{Var}\left[\varphi(W,\theta,\delta,\eta^{(1)},\eta^{(0)})\right]}{n}
		 	\]
		 	 in probability as $n\rightarrow\infty$. 
		 	\label{thm 2}
		 \end{theorem}
	 
	     \begin{corollary}
	     \label{cor:aym std nor}
	     Under the assumptions and  regularity conditions of Theorem \ref{thm 1}, 
	     	$\left({\widehat{\sigma}_{\text{mdel}}}\right)^{-1}\left(\widehat{\theta}_{\text{mdel}}-\theta\right)$ is asymptotically standard normal.
	     \end{corollary}
Corollary \ref{cor:aym std nor} leads to a $100(1-\alpha)\%$ Wald confidence interval of $\theta$: $$\operatorname{CI}:=\left(\widehat{\theta}_{\text{mdel}}-z_{\alpha/2}{{\widehat{\sigma}_{\text{mdel}}}},\widehat{\theta}_{\text{mdel}}+z_{\alpha/2}{{\widehat{\sigma}_{\text{mdel}}}}\right),$$ where $z_{\alpha/2}$ is the upper quantile of the standard normal distribution.\par
As we mentioned in Section \ref{subsec:tea}, Zhang's approach with the covariate-outcome relationship estimated by an ML method seriously underestimates the variance. In contrast, our proposed EL approach recovers the true variance. To illustrate this point, we conduct a simulation study following the setting of \citet{Wager2016}. The setting is a special case of the simulation studies in Section \ref{subsec:sim} with coefficients equal to $(1,0,\cdots,0)$ or a permutation of $(1,\frac{1}{2},\cdots,\frac{1}{p})$. And for both Zhang's and our proposed MDEL approach, the covariate-outcome relationship is modelled by Lasso.\par

\begin{figure}[H]
	\centering
	\begin{subfigure}[t]{0.48\textwidth}
	\includegraphics[width=1\textwidth]{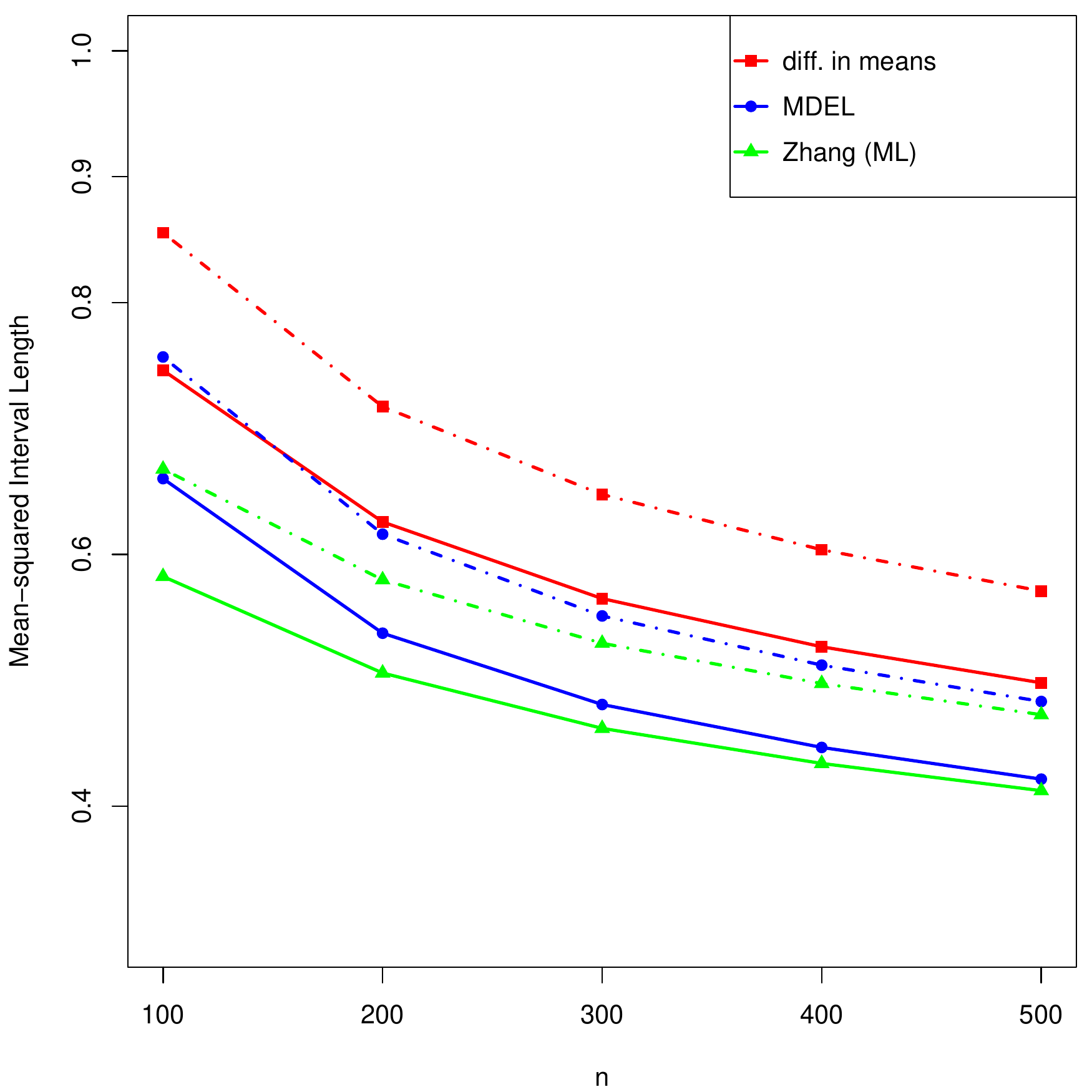}
	\caption{Mean-squared length of confidence intervals.}
	\end{subfigure}
	\begin{subfigure}[t]{0.48\textwidth}
	\includegraphics[width=1\textwidth]{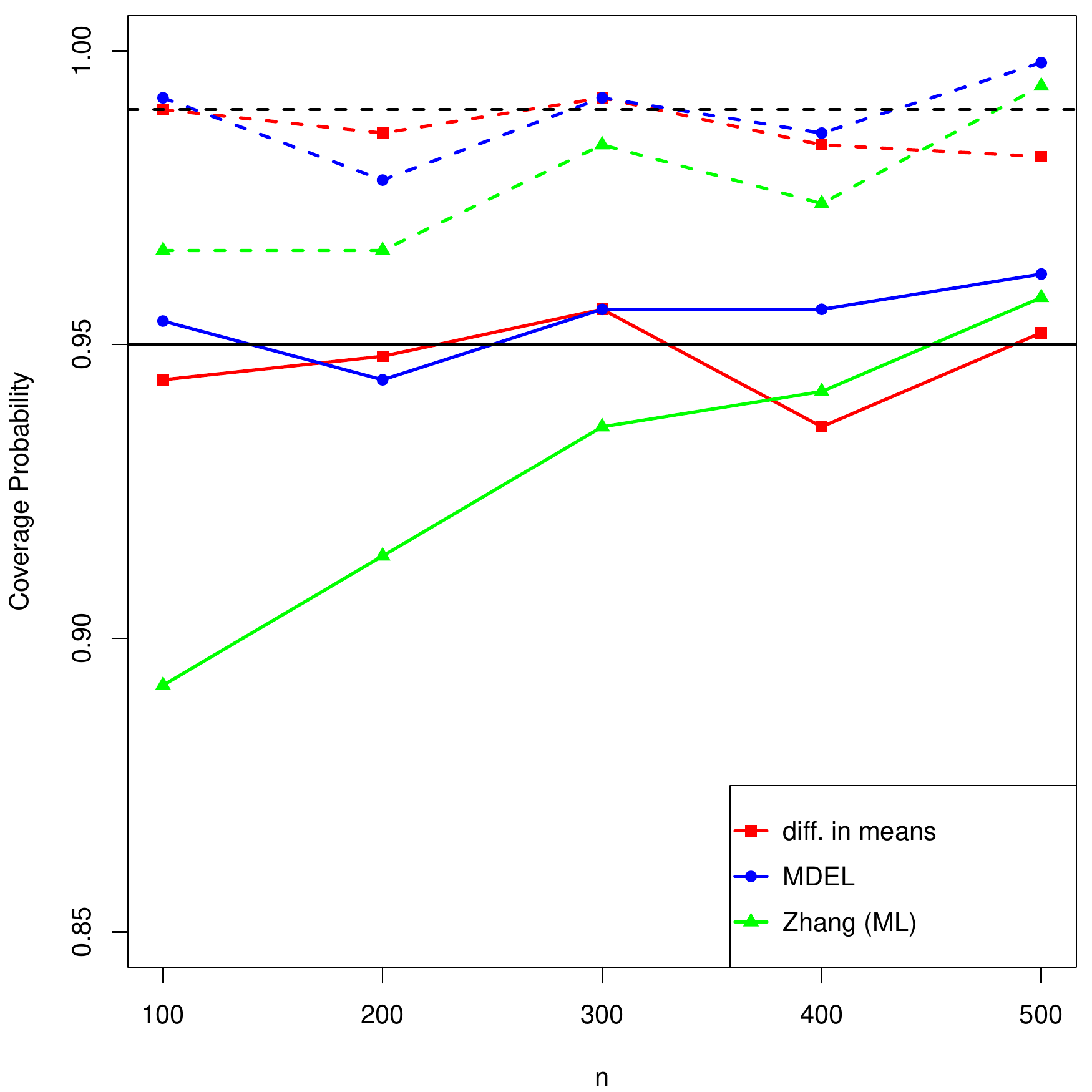}
	\caption{Coverage probability of confidence intervals.}
	\end{subfigure}
	\caption{Simulation results based on 500 Monte Carlo replications with $\beta^{(1)}=\beta^{(0)}=(1,0,\cdots,0)$, $p=500$, $\rho=0$, $\delta=0.5$ and sample size $n$ ranging from $100$ to $500$ under a simple setting described in section \ref{subsec:sim}. In the left panel, solid lines depict the mean-squared lengths of $95\%$ Wald confidence intervals and dashed-dotted lines depict the mean squared lengths of $99\%$ Wald confidence intervals. In the right panel, solid lines depict coverage proportions of $95\%$ Wald confidence intervals that cover the true $\theta$ and dashed lines depict coverage proportions of $99\%$ Wald confidence intervals that cover the true $\theta$.}
	\label{fig:S1ELC}
\end{figure}
		
\begin{figure}[H]
	\centering
	\begin{subfigure}[t]{0.48\textwidth}
	\includegraphics[width=1\textwidth]{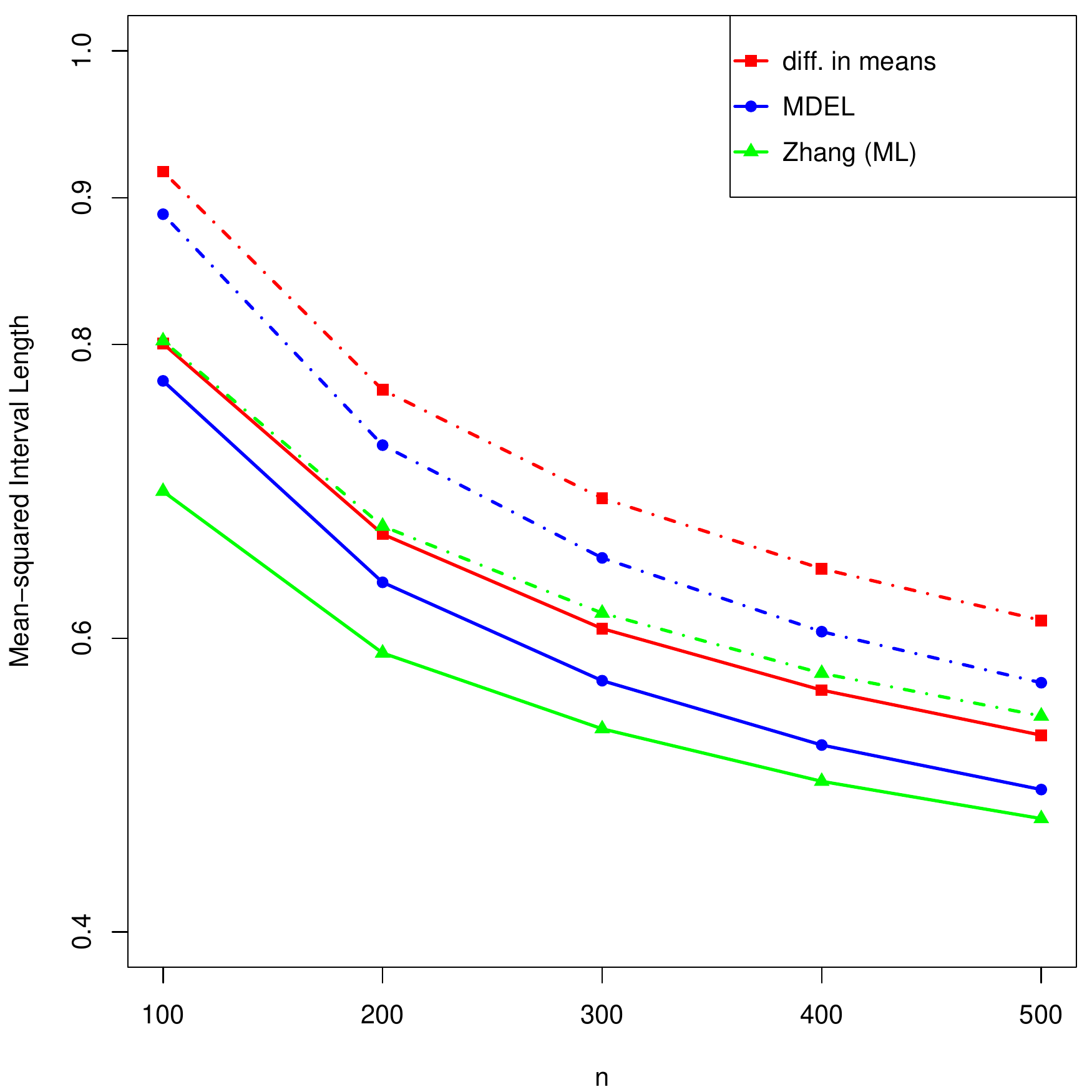}
	\caption{Mean-squared length of confidence intervals.}
	\end{subfigure}
	\begin{subfigure}[t]{0.48\textwidth}
	\includegraphics[width=1\textwidth]{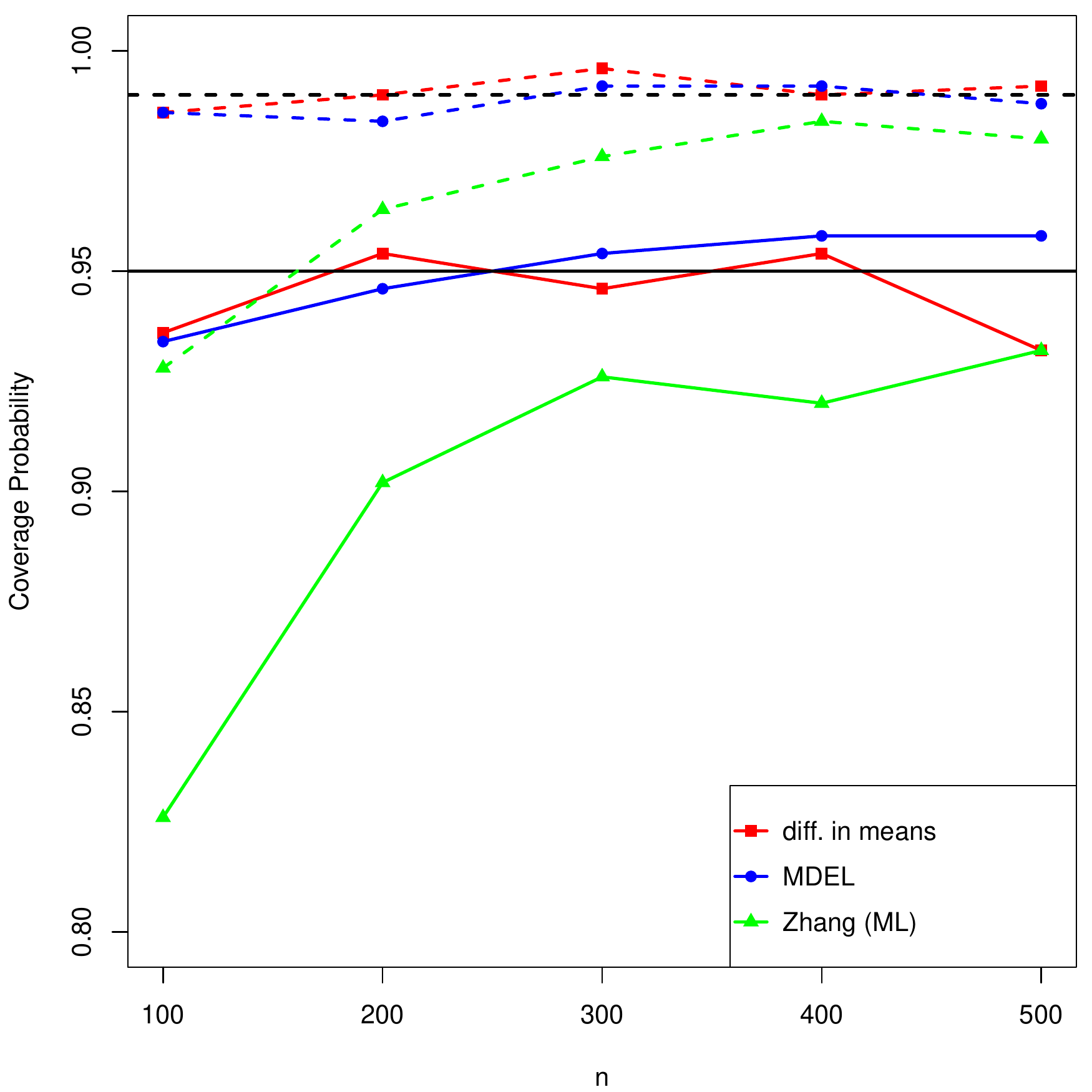}
	\caption{Coverage probability of confidence intervals.}
	\end{subfigure}
	\caption{Simulation results based on 500 Monte Carlo replications with $\beta^{(d)}$ equal to a random permutation of $(\frac{1}{1},\frac{1}{2},\cdots,\frac{1}{p})$ for $d=0,1$, $p=500$, $\rho=0.5$, $\delta=0.5$ and sample size $n$ ranging from $100$ to $500$ under the setting described in Section \ref{subsec:sim}. In the left panel, solid lines depict the mean-squared length of $95\%$ Wald confidence intervals and dashed-dotted lines depict the mean-squared length of $99\%$ Wald confidence intervals. In the right panel, solid lines depict coverage proportions of $95\%$ Wald confidence intervals that cover the true $\theta$ and dashed lines depict coverage proportions of $99\%$ Wald confidence intervals that cover the true $\theta$.}
	\label{fig:S3ELC}
\end{figure}

In Figure \ref{fig:S1ELC}, where the true signal is very sparse, the 95\% confidence intervals and 99\% confidence intervals of our EL approach are shorter than the corresponding confidence intervals based on $\widehat{\theta}_{\text{dim}}$. The coverage probability of the 95\% or 99\% confidence intervals of Zhang's EL approach is substantially lower than the true level for $n=100,200$. In contrast, the coverage probability of our EL approach resembles the nominal level for any sample size.\par

In Figure \ref{fig:S3ELC}, where the true signal is geometric, the 95\% confidence intervals and 99\% confidence intervals of our EL approach are still shorter than the corresponding ones based on $\widehat{\theta}_{\text{dim}}$. The coverage probability of the 95\% or 99\% confidence intervals of Zhang's approach is significantly below the nominal level for $n$ ranging from 100 to 500. In contrast, the coverage probability of our EL approach are very close to the nominal level for any sample size.\par
In summary, when the true signal is either $(1,0,\cdots,0)$ or a permutation of $(1,\frac{1}{2},\cdots,\frac{1}{p})$, Zhang's approach is not desirable as his proposed confidence intervals fail to cover $\theta$ in reasonable proportions whereas our approach recovers true variance and the coverage probabilities are close to the nominal levels.

\section{Practical Implementations}\label{sec:pract}
In our simulation studies and real data analysis, we utilize three popular ML methods to estimate $\eta^{(d)}$. One is Lasso \citep{Tibshirani1996} and the second one is SCAD \citep{Fan2001}. Lasso and SCAD are both penalized regression methods. Generally, they both lead to sparse solutions and thus work well for variable selection purpose. However, compared with Lasso, large coefficients would not be shrunken by SCAD and some small coefficients cannot survive after punishment. Therefore, SCAD works better for models with strong and sparse signals. The third method we use is random forests \citep{Breiman2001}, which are increasingly popular in recent years because of its flexibility and outstanding prediction ability for real complex data.\par
For Lasso, the penalty parameter $\lambda$ is determined by cross-validation criterion using \texttt{cv.glmnet} in R-package \texttt{glmnet} with 10-folds in this paper.
For SCAD, the first tuning parameter $a$ is chosen to be default 3.7 and penalty parameter $\lambda$ is determined by cross-validation criterion using \texttt{cv.ncvreg} in R-package \texttt{ncvreg} with 10-folds.
For random forests, we build 500 regression trees using \texttt{ranger} in R-package \texttt{ranger}, a fast implementation of random forests
for high dimensional data in C++ and R, with parameters set to be default.\par

\section{Simulation Studies and Real Data Analysis}\label{sec:rda and sim}

\subsection{Simulation Studies}\label{subsec:sim}
We consider linear models for $\eta^{(d)}$ with the dimension of covariates $p$ larger than the sample size $n$. The universal settings of our simulations are as follows. The covariates $X_i,i=1,\cdots,n$ are independent and identically generated from multivariate Gaussian $\mathcal{N}(1_p,\Sigma)$, where $1_p=(1,\cdots,1)^\tau$ is a $p$-dimensional vector. The assignment probability is fixed to be $\delta=0.5$ and $D_i\overset{i.i.d}{\sim}\text{Bernoulli}(\delta)$. The outcome $Y_i$ of the i-th unit under treatment $D_i=d_i$ are generated from $\mathcal{N}(X^\tau_i\beta^{(d_i)}+5I(d_i=1), 1)$, $i=1,\cdots,n$. We consider four different size scales 
\[
(n,p)=\{(80,200),(160,200),(200,1000),(800,1000)\}.
\]
Define $0^0=1$, signals and the covariance matrix of the covariates, $\mathbf{\Sigma}$, are different as follows.
\begin{simulation}[Sufficient Sparsity]
$\beta^{(1)}_i=3\cdot1(i\leq3), \beta^{(0)}_j=2\cdot1(j\leq3)$ and $\mathbf{\Sigma}_{ij}=\rho^{1(i\neq j)}$.
\label{sim 1}
\end{simulation}

\begin{simulation}[\citet{Fan2012}]
\[
(\beta^{(d)}_i)_{i=1,2,3,5,7,11,13,17,19,23}=(1.01, -0.06, 0.72, 1.55, 2.32, -0.36, 3.75, -2.04, -0.13, 0.61)^\tau, d=0,1
\]
and $\mathbf{\Sigma}_{ij}=\rho^{|i-j|}$.
\label{sim 2}
\end{simulation}

\begin{simulation}[Dense Geometry \citep{Wager2016}]
$\beta^{(1)}_i=11^{-10i/p}, \beta^{(0)}_j=10^{-10j/p}$ and $\mathbf{\Sigma}_{ij}=\rho^{|i-j|}$.
\label{sim 3}
\end{simulation}
 
Simulation \ref{sim 1} has sparse and strong signals. Simulation \ref{sim 2} has sparse signals with more challenging cofficients. Simulation \ref{sim 3} is identical to the geometric case of \citet{Wager2016}. Results of simulations are all based on $5000$ Monte Carlo data sets and given in Table \ref{tab:sim1}, \ref{tab:sim2} and \ref{tab:sim3}. First, we summarize the results in Table \ref{tab:sim1} and \ref{tab:sim2} (sparse case):
     \begin{itemize}
     	\item[(a)] Compared with the simple approach of difference in means, the EL estimators with any outcome model have significantly smaller SDs and RMSEs.
     	\item[(b)] Among the EL estimators with one outcome model, the estimators using SCAD perform relatively better than other estimators, and estimators using random forests perform worst in sense of RMSE. As expected, the EL estimators with multiple models perform closest to those with SCAD, and better than all other estimators when $\rho=0.5$ and $(n,p)=(80,200)$ and when $(n,p)=(800,1000)$.
     	\item[(c)] Using SCAD to model the covariate-outcome relationship, the EL estimators perform similarly to Wager's estimators in terms of SD and RMSE. However, when Lasso or random forests model are adopted, the EL estimators outperform Wager's estimators. In Simulation \ref{sim 1}, compared with Wager's estimators with random forests, the EL estimators with random forests have an average of $15.9\%$ reduction in RMSE for $\rho=0$ and $6.3\%$ reduction in RMSE for $\rho=0.5$. In simulation \ref{sim 2}, compared with Wager's estimators with random forests, the EL estimators with random forests have an average of $14.5\%$ reduction in RMSE for $\rho=0$ and $18.4\%$ reduction in RMSE for $\rho=0.5$. The reduction is more obvious when $n$ is larger.
     	\item[(d)] The SEs of the EL estimators with one outcome model are very close to their corresponding SDs, and the coverage probabilities of the EL estimators with one model are close to the nominal levels. However, the variances of the EL estimators with multiple models are slightly overestimated, but in a reasonable range.
     \end{itemize}
Results in Table \ref{tab:sim3} are summarized as follows (dense case):
     \begin{itemize}
     	\item[(a)] When $\rho=0$ and $n$ is small compared to the dimension of covariates($(n,p)=(80,200),(200,1000)$), compared with the simple approach of difference in means, there is no significant reduction in RMSE for the EL estimators. Otheriwse, compared with the difference in means estimators, the EL estimators with any outcome model have significantly smaller SDs and RMSEs.
     	\item[(b)] When $\rho=0$ and $n$ is small compared to the dimension of covariates, there is no significant difference among different estimators. In other cases, among the EL estimators using one outcome model, the estimators with Lasso generally perform best and the estimators with random forests perform worst in sense of RMSE. As expected, the EL estimators with multiple ML models perform closest to the ones with the best model.
     	\item[(c)] Under Lasso or SCAD model, the EL estimators perform similarly to Wager's estimators in terms of RMSE. However, under random forests model, the EL estimators significantly outperform Wager's estimators. For example, when $(n,p)=(160,200)$ and $\rho=0.5$, the RMSE of $\widehat{\theta}_{\text{mdel}}$(RF) is $0.304$ while the RMSE of $\widehat{\theta}_{\text{wdtt}}$(RF) is $0.446$.
     	\item[(d)] When $\rho=0$ and $(n,p)=(80,200)$, the variances of EL estimators are underestimated and the coverage rates of the EL estimators are smaller than the nominal levels, but still in reasonable range. 

     \end{itemize}

\begin{table}[H]
     	\caption{Results of Simulation \ref{sim 1} based on 5000 Monte Carlo replications}
     	\resizebox{\linewidth}{!}{
     		\begin{threeparttable}
            \centering
     				\begin{tabular}{lrrrrrrrrrrrrrr}
    \toprule
    &       & \multicolumn{6}{l}{$\rho=0$}                     &       & \multicolumn{6}{l}{$\rho=0.5$} \\
     					\cmidrule{3-8}\cmidrule{10-15}    Estimator &       & \multicolumn{1}{l}{Bias} & \multicolumn{1}{l}{SD} & \multicolumn{1}{l}{SE} & \multicolumn{1}{l}{RMSE} & \multicolumn{1}{l}{Cov95} & \multicolumn{1}{l}{Cov99} &       & \multicolumn{1}{l}{Bias} & \multicolumn{1}{l}{SD} & \multicolumn{1}{l}{SE} & \multicolumn{1}{l}{RMSE} & \multicolumn{1}{l}{Cov95} & \multicolumn{1}{l}{Cov99} \\
     					\midrule
     					\multicolumn{15}{l}{$(n,p)=(80,200)$} \\
    $\widehat{\theta}_{\text{dim}}$ &       & -0.004 & 1.002 & 1.015 & 1.002 & 0.950 & 0.988 &       & -0.012 & 1.389 & 1.417 & 1.389 & 0.953 & 0.991 \\
    $\widehat{\theta}_{\text{wdtt}}$(LASSO) &       & 0.000 & 0.386 & 0.383 & 0.386 & 0.943 & 0.987 &       & 0.000 & 0.407 & 0.409 & 0.407 & 0.951 & 0.991 \\
    $\widehat{\theta}_{\text{wdtt}}$(SCAD) &       & 0.001 & 0.318 & 0.315 & 0.318 & 0.948 & 0.989 &       & 0.002 & 0.381 & 0.378 & 0.381 & 0.947 & 0.988 \\
    $\widehat{\theta}_{\text{wdtt}}$(RF) &       & -0.005 & 0.959 & 0.970 & 0.959 & 0.947 & 0.988 &       & -0.010 & 0.820 & 0.841 & 0.820 & 0.955 & 0.990 \\
    $\widehat{\theta}_{\text{mdel}}$(LASSO) &       & 0.001 & 0.349 & 0.349 & 0.349 & 0.949 & 0.987 &       & 0.000 & 0.399 & 0.406 & 0.399 & 0.952 & 0.991 \\
    $\widehat{\theta}_{\text{mdel}}$(SCAD) &       & 0.001 & 0.318 & 0.316 & 0.318 & 0.945 & 0.990 &       & 0.002 & 0.382 & 0.381 & 0.382 & 0.948 & 0.988 \\
    $\widehat{\theta}_{\text{mdel}}$(RF) &       & -0.008 & 0.945 & 0.948 & 0.945 & 0.948 & 0.988 &       & -0.008 & 0.743 & 0.772 & 0.743 & 0.956 & 0.992 \\
    $\widehat{\theta}_{\text{mdel}}$(MULTI) &       & 0.003 & 0.321 & 0.341 & 0.321 & 0.959 & 0.992 &       & 0.002 & 0.374 & 0.397 & 0.374 & 0.957 & 0.993 \\
    \midrule
    \multicolumn{15}{l}{$(n,p)=(160,200)$} \\
    $\widehat{\theta}_{\text{dim}}$ &       & -0.004 & 0.713 & 0.718 & 0.713 & 0.950 & 0.990 &       & -0.016 & 0.999 & 1.003 & 0.999 & 0.951 & 0.991 \\
    $\widehat{\theta}_{\text{wdtt}}$(LASSO) &       & -0.001 & 0.225 & 0.228 & 0.225 & 0.953 & 0.989 &       & -0.002 & 0.261 & 0.263 & 0.261 & 0.946 & 0.992 \\
    $\widehat{\theta}_{\text{wdtt}}$(SCAD) &       & -0.001 & 0.210 & 0.212 & 0.210 & 0.948 & 0.989 &       & -0.002 & 0.250 & 0.251 & 0.250 & 0.949 & 0.990 \\
    $\widehat{\theta}_{\text{wdtt}}$(RF) &       & -0.004 & 0.653 & 0.658 & 0.653 & 0.951 & 0.990 &       & -0.005 & 0.550 & 0.556 & 0.550 & 0.951 & 0.990 \\
    $\widehat{\theta}_{\text{mdel}}$(LASSO) &       & -0.001 & 0.215 & 0.218 & 0.215 & 0.950 & 0.990 &       & -0.001 & 0.259 & 0.262 & 0.259 & 0.950 & 0.992 \\
    $\widehat{\theta}_{\text{mdel}}$(SCAD) &       & -0.001 & 0.210 & 0.212 & 0.210 & 0.950 & 0.989 &       & -0.002 & 0.250 & 0.252 & 0.250 & 0.952 & 0.991 \\
    $\widehat{\theta}_{\text{mdel}}$(RF) &       & -0.004 & 0.536 & 0.560 & 0.536 & 0.959 & 0.992 &       & -0.001 & 0.510 & 0.519 & 0.510 & 0.952 & 0.991 \\
    $\widehat{\theta}_{\text{mdel}}$(MULTI) &       & -0.001 & 0.212 & 0.222 & 0.212 & 0.956 & 0.991 &       & -0.002 & 0.251 & 0.261 & 0.251 & 0.959 & 0.991 \\
    \midrule
    \multicolumn{15}{l}{$(n,p)=(200,1000)$} \\
    $\widehat{\theta}_{\text{dim}}$ &       & -0.008 & 0.642 & 0.642 & 0.642 & 0.946 & 0.989 &       & -0.019 & 0.895 & 0.896 & 0.895 & 0.947 & 0.990 \\
    $\widehat{\theta}_{\text{wdtt}}$(LASSO) &       & 0.002 & 0.206 & 0.207 & 0.206 & 0.950 & 0.988 &       & 0.002 & 0.237 & 0.238 & 0.237 & 0.950 & 0.989 \\
    $\widehat{\theta}_{\text{wdtt}}$(SCAD) &       & 0.003 & 0.190 & 0.189 & 0.190 & 0.948 & 0.990 &       & 0.003 & 0.225 & 0.225 & 0.225 & 0.951 & 0.989 \\
    $\widehat{\theta}_{\text{wdtt}}$(RF) &       & -0.007 & 0.618 & 0.618 & 0.618 & 0.946 & 0.990 &       & -0.009 & 0.510 & 0.508 & 0.510 & 0.947 & 0.988 \\
    $\widehat{\theta}_{\text{mdel}}$(LASSO) &       & 0.004 & 0.194 & 0.194 & 0.194 & 0.949 & 0.990 &       & 0.003 & 0.235 & 0.236 & 0.235 & 0.951 & 0.989 \\
    $\widehat{\theta}_{\text{mdel}}$(SCAD) &       & 0.003 & 0.190 & 0.190 & 0.190 & 0.950 & 0.990 &       & 0.003 & 0.225 & 0.225 & 0.225 & 0.952 & 0.990 \\
    $\widehat{\theta}_{\text{mdel}}$(RF) &       & -0.005 & 0.565 & 0.576 & 0.565 & 0.953 & 0.991 &       & -0.006 & 0.483 & 0.482 & 0.483 & 0.950 & 0.987 \\
    $\widehat{\theta}_{\text{mdel}}$(MULTI) &       & 0.003 & 0.191 & 0.196 & 0.191 & 0.955 & 0.992 &       & 0.003 & 0.226 & 0.231 & 0.226 & 0.954 & 0.991 \\
    \midrule
    \multicolumn{15}{l}{$(n,p)=(800,1000)$} \\
    $\widehat{\theta}_{\text{dim}}$ &       & -0.005 & 0.321 & 0.320 & 0.321 & 0.949 & 0.990 &       & -0.004 & 0.450 & 0.447 & 0.450 & 0.951 & 0.991 \\
    $\widehat{\theta}_{\text{wdtt}}$(LASSO) &       & -0.001 & 0.096 & 0.096 & 0.096 & 0.949 & 0.991 &       & -0.002 & 0.114 & 0.113 & 0.114 & 0.949 & 0.989 \\
    $\widehat{\theta}_{\text{wdtt}}$(SCAD) &       & -0.001 & 0.093 & 0.094 & 0.093 & 0.950 & 0.991 &       & -0.002 & 0.112 & 0.112 & 0.112 & 0.949 & 0.989 \\
    $\widehat{\theta}_{\text{wdtt}}$(RF) &       & -0.004 & 0.294 & 0.293 & 0.294 & 0.951 & 0.991 &       & -0.005 & 0.238 & 0.237 & 0.238 & 0.949 & 0.989 \\
    $\widehat{\theta}_{\text{mdel}}$(LASSO) &       & -0.001 & 0.094 & 0.094 & 0.094 & 0.950 & 0.991 &       & -0.002 & 0.113 & 0.113 & 0.113 & 0.949 & 0.989 \\
    $\widehat{\theta}_{\text{mdel}}$(SCAD) &       & -0.001 & 0.093 & 0.094 & 0.093 & 0.951 & 0.991 &       & -0.002 & 0.112 & 0.112 & 0.112 & 0.950 & 0.989 \\
    $\widehat{\theta}_{\text{mdel}}$(RF) &       & -0.002 & 0.189 & 0.194 & 0.189 & 0.957 & 0.991 &       & -0.005 & 0.230 & 0.230 & 0.230 & 0.948 & 0.989 \\
    $\widehat{\theta}_{\text{mdel}}$(MULTI) &       & -0.001 & 0.093 & 0.095 & 0.093 & 0.952 & 0.992 &       & -0.002 & 0.112 & 0.113 & 0.112 & 0.949 & 0.989 \\
    \bottomrule
    \end{tabular}%
     			\begin{tablenotes}[flushleft]
     				\small
     				\item Bias = average bias of 5000 Monte Carlo estimators, SD = sample standard deviation of estimators, SE = average of model-based standard error, RMSE = empirical root mean square error, Cov95 = proportion of $95\%$ Wald confidence intervals covering the true $\theta$, Cov99 = proportion of $99\%$ Wald confidence intervals covering the true $\theta$.
     			\end{tablenotes}
     		\end{threeparttable}
     	}
     	\label{tab:sim1}%
     \end{table}%
     
     \begin{table}[H]
     	\caption{Results of Simulation \ref{sim 2} based on 5000 Monte Carlo replications}
     	\resizebox{\linewidth}{!}{
     		\begin{threeparttable}
     			\centering

     				\begin{tabular}{lrrrrrrrrrrrrrr}
     					\toprule
     					&       & \multicolumn{6}{l}{$\rho=0$}                     &       & \multicolumn{6}{l}{$\rho=0.5$} \\
     					\cmidrule{3-8}\cmidrule{10-15}    Estimator &       & \multicolumn{1}{l}{Bias} & \multicolumn{1}{l}{SD} & \multicolumn{1}{l}{SE} & \multicolumn{1}{l}{RMSE} & \multicolumn{1}{l}{Cov95} & \multicolumn{1}{l}{Cov99} &       & \multicolumn{1}{l}{Bias} & \multicolumn{1}{l}{SD} & \multicolumn{1}{l}{SE} & \multicolumn{1}{l}{RMSE} & \multicolumn{1}{l}{Cov95} & \multicolumn{1}{l}{Cov99} \\
     					\midrule
     					\multicolumn{15}{l}{$(n,p)=(80,200)$} \\
     					$\widehat{\theta}_{\text{dim}}$ &       & 0.028 & 1.196 & 1.210 & 1.196 & 0.950 & 0.988 &       & 0.017 & 1.231 & 1.243 & 1.231 & 0.946 & 0.988 \\
     					$\widehat{\theta}_{\text{wdtt}}$(LASSO) &       & 0.010 & 0.584 & 0.586 & 0.584 & 0.951 & 0.987 &       & 0.006 & 0.536 & 0.535 & 0.536 & 0.952 & 0.990 \\
     					$\widehat{\theta}_{\text{wdtt}}$(SCAD) &       & 0.007 & 0.451 & 0.447 & 0.451 & 0.949 & 0.989 &       & 0.007 & 0.448 & 0.442 & 0.448 & 0.946 & 0.989 \\
     					$\widehat{\theta}_{\text{wdtt}}$(RF) &       & 0.026 & 1.141 & 1.155 & 1.141 & 0.950 & 0.989 &       & 0.017 & 1.153 & 1.165 & 1.153 & 0.946 & 0.988 \\
     					$\widehat{\theta}_{\text{mdel}}$(LASSO) &       & 0.008 & 0.539 & 0.541 & 0.539 & 0.953 & 0.987 &       & 0.006 & 0.494 & 0.489 & 0.494 & 0.951 & 0.989 \\
     					$\widehat{\theta}_{\text{mdel}}$(SCAD) &       & 0.008 & 0.457 & 0.453 & 0.457 & 0.950 & 0.988 &       & 0.007 & 0.457 & 0.448 & 0.457 & 0.948 & 0.989 \\
     					$\widehat{\theta}_{\text{mdel}}$(RF) &       & 0.023 & 1.130 & 1.129 & 1.131 & 0.948 & 0.987 &       & 0.013 & 1.105 & 1.112 & 1.105 & 0.951 & 0.987 \\
     					$\widehat{\theta}_{\text{mdel}}$(MULTI) &       & 0.007 & 0.457 & 0.466 & 0.457 & 0.957 & 0.990 &       & 0.006 & 0.440 & 0.448 & 0.440 & 0.956 & 0.991 \\
     					\midrule
     					\multicolumn{15}{l}{$(n,p)=(160,200)$} \\
     					$\widehat{\theta}_{\text{dim}}$ &      &-0.018 & 0.859 & 0.853 & 0.859 & 0.947 & 0.988 &       & -0.010 & 0.891 & 0.878 & 0.891 & 0.947 & 0.988 \\
     					$\widehat{\theta}_{\text{wdtt}}$(LASSO) &      &0.000 & 0.226 & 0.223 & 0.226 & 0.947 & 0.988 &       & -0.001 & 0.217 & 0.215 & 0.217 & 0.949 & 0.988 \\
                        $\widehat{\theta}_{\text{wdtt}}$(SCAD) &       &0.002 & 0.182 & 0.182 & 0.182 & 0.945 & 0.988 &       & 0.001 & 0.183 & 0.183 & 0.183 & 0.953 & 0.989 \\
    $\widehat{\theta}_{\text{wdtt}}$(RF) &      &-0.016 & 0.788 & 0.784 & 0.789 & 0.948 & 0.989 &       & -0.008 & 0.794 & 0.782 & 0.794 & 0.947 & 0.989 \\
    $\widehat{\theta}_{\text{mdel}}$(LASSO) &      &0.001 & 0.208 & 0.207 & 0.208 & 0.950 & 0.989 &       & 0.000 & 0.201 & 0.201 & 0.201 & 0.952 & 0.987 \\
    $\widehat{\theta}_{\text{mdel}}$(SCAD) &       &0.001 & 0.183 & 0.183 & 0.183 & 0.948 & 0.989 &       & 0.001 & 0.183 & 0.184 & 0.183 & 0.953 & 0.990 \\
    $\widehat{\theta}_{\text{mdel}}$(RF) &       &-0.011 & 0.655 & 0.675 & 0.655 & 0.958 & 0.992 &       & -0.004 & 0.622 & 0.638 & 0.622 & 0.955 & 0.992 \\
    $\widehat{\theta}_{\text{mdel}}$(MULTI) &       &0.001 & 0.184 & 0.200 & 0.184 & 0.960 & 0.994 &       & 0.001 & 0.184 & 0.201 & 0.184 & 0.966 & 0.993 \\
    \midrule
     					\multicolumn{15}{l}{$(n,p)=(200,1000)$} \\
     					$\widehat{\theta}_{\text{dim}}$ &       & -0.003 & 0.773 & 0.763 & 0.773 & 0.945 & 0.988 &       & -0.002 & 0.796 & 0.785 & 0.796 & 0.949 & 0.988 \\
     					$\widehat{\theta}_{\text{wdtt}}$(LASSO) &       & 0.001 & 0.226 & 0.222 & 0.226 & 0.945 & 0.991 &       & 0.001 & 0.212 & 0.209 & 0.212 & 0.945 & 0.990 \\
     					$\widehat{\theta}_{\text{wdtt}}$(SCAD) &       & 0.002 & 0.164 & 0.163 & 0.164 & 0.950 & 0.990 &       & 0.003 & 0.166 & 0.164 & 0.166 & 0.948 & 0.989 \\
     					$\widehat{\theta}_{\text{wdtt}}$(RF) &       & -0.003 & 0.746 & 0.736 & 0.746 & 0.946 & 0.988 &       & -0.002 & 0.758 & 0.748 & 0.758 & 0.949 & 0.987 \\
     					$\widehat{\theta}_{\text{mdel}}$(LASSO) &       & 0.001 & 0.199 & 0.198 & 0.199 & 0.949 & 0.992 &       & 0.002 & 0.185 & 0.185 & 0.185 & 0.951 & 0.991 \\
     					$\widehat{\theta}_{\text{mdel}}$(SCAD) &       & 0.002 & 0.164 & 0.164 & 0.164 & 0.952 & 0.990 &       & 0.003 & 0.167 & 0.165 & 0.167 & 0.948 & 0.990 \\
     					$\widehat{\theta}_{\text{mdel}}$(RF) &       & -0.003 & 0.695 & 0.696 & 0.695 & 0.949 & 0.989 &       & -0.006 & 0.671 & 0.677 & 0.671 & 0.949 & 0.991 \\
     					$\widehat{\theta}_{\text{mdel}}$(MULTI) &       & 0.002 & 0.165 & 0.175 & 0.165 & 0.965 & 0.995 &       & 0.003 & 0.168 & 0.177 & 0.168 & 0.962 & 0.993 \\
     					\midrule
     					\multicolumn{15}{l}{$(n,p)=(800,1000)$} \\
     					$\widehat{\theta}_{\text{dim}}$ &       &0.010 & 0.385 & 0.381 & 0.385 & 0.949 & 0.990 &       & 0.009 & 0.396 & 0.392 & 0.397 & 0.949 & 0.989 \\
    $\widehat{\theta}_{\text{wdtt}}$(LASSO) &       &0.000 & 0.076 & 0.077 & 0.076 & 0.952 & 0.991 &       & 0.000 & 0.076 & 0.076 & 0.076 & 0.947 & 0.991 \\
    $\widehat{\theta}_{\text{wdtt}}$(SCAD) &       &-0.001 & 0.072 & 0.073 & 0.072 & 0.953 & 0.992 &       & -0.001 & 0.072 & 0.073 & 0.072 & 0.955 & 0.991 \\
    $\widehat{\theta}_{\text{wdtt}}$(RF) &      &0.009 & 0.354 & 0.350 & 0.354 & 0.947 & 0.990 &       & 0.009 & 0.351 & 0.348 & 0.352 & 0.946 & 0.990 \\
    $\widehat{\theta}_{\text{mdel}}$(LASSO) &       &0.000 & 0.073 & 0.074 & 0.073 & 0.954 & 0.991 &       & 0.000 & 0.073 & 0.074 & 0.073 & 0.954 & 0.992 \\
    $\widehat{\theta}_{\text{mdel}}$(SCAD) &       &-0.001 & 0.072 & 0.073 & 0.072 & 0.954 & 0.992 &       & -0.001 & 0.072 & 0.073 & 0.072 & 0.956 & 0.991 \\
    $\widehat{\theta}_{\text{mdel}}$(RF) &       &0.003 & 0.236 & 0.237 & 0.236 & 0.955 & 0.991 &       & 0.005 & 0.224 & 0.226 & 0.224 & 0.953 & 0.989 \\
    $\widehat{\theta}_{\text{mdel}}$(MULTI) &       &-0.001 & 0.072 & 0.075 & 0.072 & 0.958 & 0.992 &       & -0.001 & 0.072 & 0.075 & 0.072 & 0.958 & 0.992 \\
    \bottomrule
     				\end{tabular}%
     			
     			\begin{tablenotes}[flushleft]
     				\small
     				\item Bias = average bias of 5000 Monte Carlo estimators, SD = sample standard deviation of estimators, SE = average of model-based standard error, RMSE = empirical root mean square error, Cov95 = proportion of $95\%$ Wald confidence intervals covering the true $\theta$, Cov99 = proportion of $99\%$ Wald confidence intervals covering the true $\theta$.
     			\end{tablenotes}
     		\end{threeparttable}
     	}
     	\label{tab:sim2}%
     \end{table}%

     \begin{table}[H]
     	\caption{Results of Simulation \ref{sim 3} based on 5000 Monte Carlo replications}
     	\resizebox{\linewidth}{!}{
     		\begin{threeparttable}
     			\centering
     				\begin{tabular}{lrrrrrrrrrrrrrr}
     					\toprule
     					&       & \multicolumn{6}{l}{$\rho=0$}                     &       & \multicolumn{6}{l}{$\rho=0.5$} \\
     					\cmidrule{3-8}\cmidrule{10-15}    Estimator &       & \multicolumn{1}{l}{Bias} & \multicolumn{1}{l}{SD} & \multicolumn{1}{l}{SE} & \multicolumn{1}{l}{RMSE} & \multicolumn{1}{l}{Cov95} & \multicolumn{1}{l}{Cov99} &       & \multicolumn{1}{l}{Bias} & \multicolumn{1}{l}{SD} & \multicolumn{1}{l}{SE} & \multicolumn{1}{l}{RMSE} & \multicolumn{1}{l}{Cov95} & \multicolumn{1}{l}{Cov99} \\
     					\midrule
     					\multicolumn{15}{l}{$(n,p)=(80,200)$} \\
     					$\widehat{\theta}_{\text{dim}}$ &       & 0.003 & 0.489 & 0.490 & 0.489 & 0.949 & 0.990 &       & 0.002 & 0.732 & 0.737 & 0.732 & 0.949 & 0.990 \\
     					$\widehat{\theta}_{\text{wdtt}}$(LASSO) &       & 0.002 & 0.453 & 0.451 & 0.453 & 0.946 & 0.987 &       & 0.000 & 0.435 & 0.432 & 0.435 & 0.945 & 0.985 \\
     					$\widehat{\theta}_{\text{wdtt}}$(SCAD) &       & 0.004 & 0.458 & 0.460 & 0.458 & 0.946 & 0.988 &       & 0.002 & 0.497 & 0.497 & 0.497 & 0.945 & 0.984 \\
     					$\widehat{\theta}_{\text{wdtt}}$(RF) &       & 0.002 & 0.478 & 0.478 & 0.478 & 0.947 & 0.991 &       & 0.000 & 0.667 & 0.671 & 0.667 & 0.948 & 0.990 \\
     					$\widehat{\theta}_{\text{mdel}}$(LASSO) &       & 0.003 & 0.464 & 0.452 & 0.464 & 0.943 & 0.987 &       & 0.000 & 0.425 & 0.422 & 0.425 & 0.945 & 0.988 \\
     					$\widehat{\theta}_{\text{mdel}}$(SCAD) &       & 0.003 & 0.469 & 0.459 & 0.469 & 0.940 & 0.987 &       & 0.003 & 0.510 & 0.507 & 0.510 & 0.946 & 0.984 \\
     					$\widehat{\theta}_{\text{mdel}}$(RF) &       & 0.003 & 0.489 & 0.478 & 0.489 & 0.940 & 0.987 &       & -0.002 & 0.583 & 0.599 & 0.582 & 0.952 & 0.992 \\
     					$\widehat{\theta}_{\text{mdel}}$(MULTI) &       & 0.005 & 0.470 & 0.446 & 0.470 & 0.932 & 0.983 &       & 0.001 & 0.433 & 0.425 & 0.433 & 0.943 & 0.986 \\
     					\midrule
     					\multicolumn{15}{l}{$(n,p)=(160,200)$} \\
     					$\widehat{\theta}_{\text{dim}}$ &       &-0.002 & 0.350 & 0.346 & 0.350 & 0.947 & 0.989 &       & -0.005 & 0.527 & 0.521 & 0.527 & 0.948 & 0.989 \\
    $\widehat{\theta}_{\text{wdtt}}$(LASSO) &       &-0.001 & 0.259 & 0.255 & 0.259 & 0.950 & 0.989 &       & -0.002 & 0.220 & 0.219 & 0.220 & 0.948 & 0.990 \\
    $\widehat{\theta}_{\text{wdtt}}$(SCAD) &       &0.000 & 0.260 & 0.256 & 0.260 & 0.949 & 0.988 &       & -0.001 & 0.257 & 0.254 & 0.257 & 0.949 & 0.990 \\
    $\widehat{\theta}_{\text{wdtt}}$(RF) &       &-0.002 & 0.336 & 0.332 & 0.336 & 0.946 & 0.990 &       & -0.003 & 0.447 & 0.442 & 0.446 & 0.948 & 0.989 \\
    $\widehat{\theta}_{\text{mdel}}$(LASSO) &       &-0.001 & 0.260 & 0.255 & 0.260 & 0.946 & 0.989 &       & -0.002 & 0.211 & 0.211 & 0.211 & 0.952 & 0.989 \\
    $\widehat{\theta}_{\text{mdel}}$(SCAD) &       &0.000 & 0.266 & 0.261 & 0.266 & 0.947 & 0.988 &       & -0.001 & 0.262 & 0.260 & 0.262 & 0.951 & 0.991 \\
    $\widehat{\theta}_{\text{mdel}}$(RF) &       &-0.003 & 0.326 & 0.322 & 0.326 & 0.945 & 0.988 &       & 0.000 & 0.304 & 0.317 & 0.304 & 0.959 & 0.992 \\
    $\widehat{\theta}_{\text{mdel}}$(MULTI) &       &0.000 & 0.258 & 0.250 & 0.258 & 0.942 & 0.988 &       & -0.002 & 0.212 & 0.215 & 0.212 & 0.952 & 0.991 \\

     					\midrule
     					\multicolumn{15}{l}{$(n,p)=(200,1000)$} \\
     					$\widehat{\theta}_{\text{dim}}$ &       & 0.001 & 0.663 & 0.660 & 0.663 & 0.945 & 0.988 &       & -0.001 & 1.113 & 1.109 & 1.113 & 0.947 & 0.988 \\
     					$\widehat{\theta}_{\text{wdtt}}$(LASSO) &       & 0.001 & 0.640 & 0.641 & 0.640 & 0.950 & 0.989 &       & 0.002 & 0.753 & 0.758 & 0.753 & 0.950 & 0.990 \\
     					$\widehat{\theta}_{\text{wdtt}}$(SCAD) &       & 0.001 & 0.645 & 0.646 & 0.645 & 0.949 & 0.989 &       & 0.003 & 0.864 & 0.867 & 0.864 & 0.950 & 0.991 \\
     					$\widehat{\theta}_{\text{wdtt}}$(RF) &       & 0.002 & 0.655 & 0.653 & 0.655 & 0.946 & 0.988 &       & 0.000 & 1.064 & 1.062 & 1.064 & 0.948 & 0.988 \\
     					$\widehat{\theta}_{\text{mdel}}$(LASSO) &       & 0.002 & 0.651 & 0.643 & 0.651 & 0.947 & 0.987 &       & -0.001 & 0.735 & 0.740 & 0.735 & 0.953 & 0.990 \\
     					$\widehat{\theta}_{\text{mdel}}$(SCAD) &       & 0.002 & 0.654 & 0.646 & 0.654 & 0.946 & 0.987 &       & 0.002 & 0.858 & 0.861 & 0.858 & 0.949 & 0.991 \\
     					$\widehat{\theta}_{\text{mdel}}$(RF) &       & 0.002 & 0.661 & 0.653 & 0.661 & 0.946 & 0.987 &       & 0.001 & 0.959 & 0.978 & 0.958 & 0.952 & 0.991 \\
     					$\widehat{\theta}_{\text{mdel}}$(MULTI) &       & 0.001 & 0.653 & 0.640 & 0.653 & 0.945 & 0.985 &       & -0.002 & 0.734 & 0.737 & 0.734 & 0.952 & 0.989 \\
     					\midrule
     					\multicolumn{15}{l}{$(n,p)=(800,1000)$} \\
     					$\widehat{\theta}_{\text{dim}}$ &       &0.004 & 0.330 & 0.330 & 0.330 & 0.951 & 0.991 &       & 0.008 & 0.553 & 0.555 & 0.553 & 0.950 & 0.991 \\
    $\widehat{\theta}_{\text{wdtt}}$(LASSO) &       &0.001 & 0.155 & 0.155 & 0.155 & 0.954 & 0.990 &       & 0.001 & 0.111 & 0.111 & 0.111 & 0.953 & 0.990 \\
    $\widehat{\theta}_{\text{wdtt}}$(SCAD) &       &0.000 & 0.138 & 0.139 & 0.138 & 0.954 & 0.992 &       & -0.001 & 0.165 & 0.167 & 0.165 & 0.950 & 0.993 \\
    $\widehat{\theta}_{\text{wdtt}}$(RF) &       &0.004 & 0.320 & 0.321 & 0.320 & 0.951 & 0.991 &       & 0.008 & 0.495 & 0.495 & 0.495 & 0.949 & 0.991 \\
    $\widehat{\theta}_{\text{mdel}}$(LASSO) &       &0.001 & 0.146 & 0.146 & 0.146 & 0.953 & 0.989 &       & 0.001 & 0.104 & 0.104 & 0.104 & 0.953 & 0.989 \\
    $\widehat{\theta}_{\text{mdel}}$(SCAD) &       &0.000 & 0.139 & 0.140 & 0.139 & 0.954 & 0.991 &       & 0.000 & 0.168 & 0.169 & 0.167 & 0.950 & 0.992 \\
    $\widehat{\theta}_{\text{mdel}}$(RF) &       &0.003 & 0.294 & 0.299 & 0.294 & 0.953 & 0.991 &       & 0.007 & 0.301 & 0.306 & 0.301 & 0.958 & 0.989 \\
    $\widehat{\theta}_{\text{mdel}}$(MULTI) &       &0.000 & 0.134 & 0.135 & 0.134 & 0.955 & 0.992 &       & 0.001 & 0.104 & 0.105 & 0.104 & 0.957 & 0.989 \\

     					\bottomrule
     				\end{tabular}%

     			\begin{tablenotes}[flushleft]
     				\small
     				\item Bias = average bias of 5000 Monte Carlo estimators, SD = sample standard deviation of estimators, SE = average of model-based standard error, RMSE = empirical root mean square error, Cov95 = proportion of $95\%$ Wald confidence intervals covering the true $\theta$, Cov99 = proportion of $99\%$ Wald confidence intervals covering the true $\theta$.
     			\end{tablenotes}
     		\end{threeparttable}
     	}
     	\label{tab:sim3}%
     \end{table}%

    \begin{figure}[H]
	\centering
	\begin{subfigure}[t]{0.48\textwidth}
	\includegraphics[width=1\textwidth]{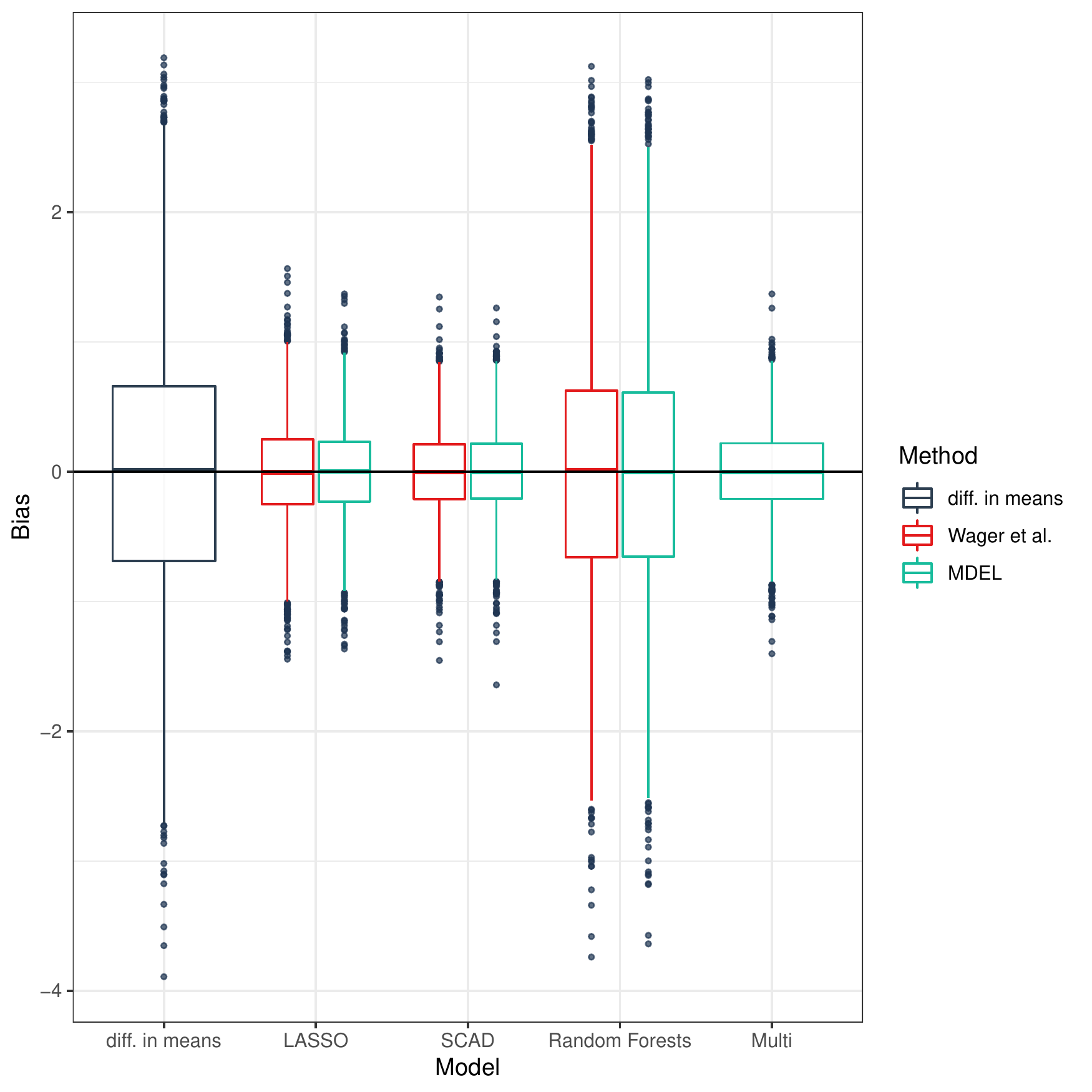}
	\caption{$(n,p)=(80,200),\rho=0$}
	\end{subfigure}
	\begin{subfigure}[t]{0.48\textwidth}
	\includegraphics[width=1\textwidth]{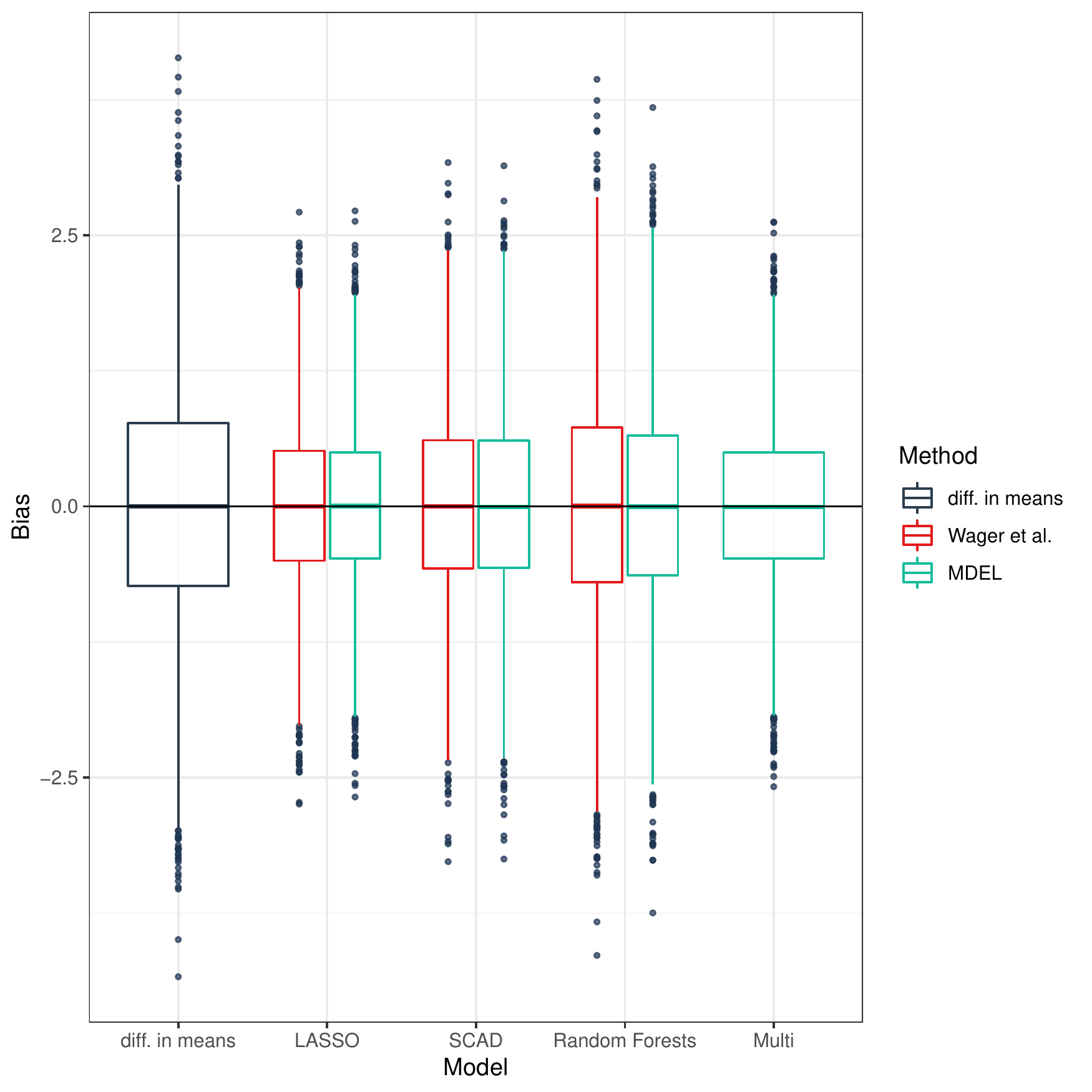}
	\caption{$(n,p)=(200,1000),\rho=0.5$}
	\end{subfigure}
     \caption{Boxplot of 5000 Monte Carlo biases based on Simulation \ref{sim 1} (in the left panel) and Simulation \ref{sim 3} (in the right panel).}
     \label{fig:boxplot}
     \end{figure}

\subsection{Analysis of ACTG 175 Data Set}\label{subsec:rda}
In this section, we apply our proposed MDEL method to data from $2139$ HIV-infected patients enrolled in AIDS Clinical Trials Group Protocol 175 (ACTG175) \citep{Hammer1996}. It is a double-blinded randomized experiment which was designed to study the treatment of patients receiving $3$ different drugs and their combinations. Patients whose CD4 cell counts from 200 to 500 per cubic millimeter were randomly assigned to  different antiretroviral regimens: zidovudine (ZDV) monotherapy, ZDV + didanosine (ddI), ZDV + zalcitabine, and ddI monotherapy. We follow the work of \citet{Tsiatis2008}, \citet{Huang2008}, and \citet{Zhang2018}, where two treatment groups are considered: patients who received ZDV monotherapy alone, with $n_0=532$ and patients who received either ZDV + ddI, or ZDV + zalcitabine, or ddI alone, with $n_1=1607$. Pre-treatment baseline covariates are $5$ continuous variables: cd40 = CD4 count(cells/mm$^3$), cd80 = CD8 count(cells/mm$^3$), age = age(years), wtkg = weight(kg), karnof = Karnofsky score(scale of 0-100), and $7$ seven binary variables: hemo = hemophilia, homo = homosexual activity, drug = history of intravenous drug use, race = race(0=white, 1=nonwhite), gender = gender(0=female), str2 = antiretroviral history(0=naive, 1=experienced), and symp = symptomatic status(0=asymptomatic).\par
In the previous work of \citet{Tsiatis2008}, \citet{Zhang2018}, and \citet{Tan2020Arxiv}, Forward-1, a forward step-wise regression model allowing for linear terms of covariates, and Forward-2, a forward ste-pwise regression model allowing for linear, quadratic and interaction terms of baseline variables, are adopted. Our proposed MDEL approach enables us to consider a much richer feature set. Therefore, we take linear and quadratic terms of continuous variables, linear and interaction terms of binary variables, and interaction terms of above two sets of coordinates as our final feature set, i.e.,
     \begin{align*}
     \mathbb{F}_{\text{ACTG175}}=\left\{\text{(cd40+cd80+age+wtkg+karnof+1)}^2\times\left(\text{(hemo+homo+drug+race+gender
     }\right.\right.\\
     \left.\left.\text{+str2+symp+1)}^2-\text{hemo}^2-\text{homo}^2-\text{drug}^2-\text{race}^2-\text{gender}^2-\text{str}^2-\text{symp}^2\right)\right\}.
     \end{align*}
This leads to $608$ explanatory variables (excluding the intercept) and we adopt Lasso, SCAD, and random forests to estimate the variable-outcome relationship using this feature set. Table \ref{tab:realdata} displays the estimates, standard errors, and confidence intervals of our proposed approach and some existing approaches described in Section \ref{sec:review} and \ref{sec:method}.\par

In practice, $\widehat{\theta}_{\text{mdel}}$(ML) denotes the point estimator of our MDEL approach, where the choices of ML include LASSO, SCAD, RF, and MULTI. Here, RF indicates the random forests method, and MULTI means we make use of multiple ML methods (LASSO, SCAD, and RF) in our MDEL method.

\begin{table}[H]
     	\centering
     	\caption{Point and interval estimates of $\theta$ for ACTG 175 data.}
     	\resizebox{\linewidth}{!}{
     		\begin{threeparttable}
     			\begin{tabular}{lllllllllll}
     				\toprule
     				Estimator &       & \multicolumn{1}{l}{Estimate} &       & \multicolumn{1}{l}{SE} &       & \multicolumn{1}{l}{Relative Efficiency} &       & 95\% Confidence Interval &       & 99\% Confidence Interval \\
     				\midrule
     				\multicolumn{11}{c}{5 folds} \\
     				$\widehat{\theta}_{\text{dim}}$ &       & 46.811 &       & 6.760 &       & 1.000 &       & ( 33.56 , 60.06 ) &       & ( 29.40 , 64.22 ) \\
     				$\widehat{\theta}_{\text{tdzl}}$(Forward-1) &       & 49.896 &       & 5.139 &       & 1.738 &       & ( 39.82 , 59.97 ) &       & ( 36.66 , 63.13 ) \\
     				$\widehat{\theta}_{\text{tdzl}}$(Forward-2) &       & 51.589 &       & 5.070 &       & 1.797 &       & ( 41.65 , 61.53 ) &       & ( 38.53 , 64.65 ) \\
     				$\widehat{\theta}_{\text{Zhang}}$(Forward-1) &       & 49.872 &       & 5.128 &       & 1.738 &       & ( 39.82 , 59.92 ) &       & ( 36.66 , 63.08 ) \\
     				$\widehat{\theta}_{\text{Zhang}}$(Forward-2) &       & 51.395 &       & 5.028 &       & 1.808 &       & ( 41.54 , 61.25 ) &       & ( 38.44 , 64.34 ) \\
     				$\widehat{\theta}_{\text{wdtt}}$(LASSO) &       & 49.785 &       & 5.233 &       & 1.669 &       & ( 39.53 , 60.04 ) &       & ( 36.31 , 63.26 ) \\
     				$\widehat{\theta}_{\text{wdtt}}$(SCAD) &       & 49.508 &       & 5.216 &       & 1.680 &       & ( 39.29 , 59.73 ) &       & ( 36.07 , 62.94 ) \\
     				$\widehat{\theta}_{\text{wdtt}}$(RF) &       & 53.442 &       & 5.253 &       & 1.656 &       & ( 43.15 , 63.74 ) &       & ( 39.91 , 66.97 ) \\
     				$\widehat{\theta}_{\text{mdel}}$(LASSO) &       & 49.938 &       & 5.200 &       & 1.690 &       & ( 39.75 , 60.13 ) &       & ( 36.54 , 63.33 ) \\
     				$\widehat{\theta}_{\text{mdel}}$(SCAD) &       & 49.483 &       & 5.197 &       & 1.692 &       & ( 39.30 , 59.67 ) &       & ( 36.10 , 62.87 ) \\
     				$\widehat{\theta}_{\text{mdel}}$(RF) &       & 53.160 &       & 5.216 &       & 1.680 &       & ( 42.94 , 63.38 ) &       & ( 39.72 , 66.60 ) \\
     				$\widehat{\theta}_{\text{mdel}}$(MULTI) &       & 50.396 &       & 5.150 &       & 1.723 &       & ( 40.30 , 60.49 ) &       & ( 37.13 , 63.66 ) \\
     				\midrule
     				\multicolumn{11}{c}{10 folds} \\
     				$\widehat{\theta}_{\text{dim}}$ &       & 46.811 &       & 6.760 &       & 1.000 &       & ( 33.56 , 60.06 ) &       & ( 29.40 , 64.22 ) \\
     				$\widehat{\theta}_{\text{tdzl}}$(Forward-1) &       & 49.896 &       & 5.139 &       & 1.738 &       & ( 39.82 , 59.97 ) &       & ( 36.66 , 63.13 ) \\
     				$\widehat{\theta}_{\text{tdzl}}$(Forward-2) &       & 51.589 &       & 5.070 &       & 1.797 &       & ( 41.65 , 61.53 ) &       & ( 38.53 , 64.65 ) \\
     				$\widehat{\theta}_{\text{Zhang}}$(Forward-1) &       & 49.872 &       & 5.128 &       & 1.738 &       & ( 39.82 , 59.92 ) &       & ( 36.66 , 63.08 ) \\
     				$\widehat{\theta}_{\text{Zhang}}$(Forward-2) &       & 51.395 &       & 5.028 &       & 1.808 &       & ( 41.54 , 61.25 ) &       & ( 38.44 , 64.34 ) \\
     				$\widehat{\theta}_{\text{wdtt}}$(LASSO) &       & 49.854 &       & 5.224 &       & 1.675 &       & ( 39.62 , 60.09 ) &       & ( 36.40 , 63.31 ) \\
     				$\widehat{\theta}_{\text{wdtt}}$(SCAD) &       & 49.991 &       & 5.210 &       & 1.684 &       & ( 39.78 , 60.20 ) &       & ( 36.57 , 63.41 ) \\
     				$\widehat{\theta}_{\text{wdtt}}$(RF) &       & 53.885 &       & 5.262 &       & 1.651 &       & ( 43.57 , 64.20 ) &       & ( 40.33 , 67.44 ) \\
     				$\widehat{\theta}_{\text{mdel}}$(LASSO) &       & 50.059 &       & 5.178 &       & 1.705 &       & ( 39.91 , 60.21 ) &       & ( 36.72 , 63.40 ) \\
     				$\widehat{\theta}_{\text{mdel}}$(SCAD) &       & 49.979 &       & 5.177 &       & 1.705 &       & ( 39.83 , 60.13 ) &       & ( 36.64 , 63.32 ) \\
     				$\widehat{\theta}_{\text{mdel}}$(RF) &       & 53.542 &       & 5.212 &       & 1.682 &       & ( 43.33 , 63.76 ) &       & ( 40.12 , 66.97 ) \\
     				$\widehat{\theta}_{\text{mdel}}$(MULTI) &       & 50.665 &       & 5.140 &       & 1.730 &       & ( 40.59 , 60.74 ) &       & ( 37.42 , 63.90 ) \\
     				\bottomrule
     			\end{tabular}%
     			\begin{tablenotes}[flushleft]
     				\small
     				\item SE = standard error, Relative Efficiency = (SE$^2$ of corresponding estimator)/(SE$^2$ of $\widehat{\theta}_{\text{dim}}$).
     			\end{tablenotes}
     		\end{threeparttable}
     	}
     	\label{tab:realdata}%
     \end{table}%
     
For inference on $\theta$, both $95\%$ and $99\%$ Wald confidence intervals are provided. The results of Table \ref{tab:realdata} give us strong evidence to reject the null hypothesis that there is no difference in treatment effect between two groups with different therapies. It is worth to note that, despite a much richer feature set with $p=608$ variables is considered, our proposed approach does not improve the estimation efficiency. This indicates that, the original explanatory variables are adequate for modeling $\eta^{(d)}(\cdot),d=0,1$. However, our data analysis result of ACTG 175 data set is still meaningful because we provide further reliability to use the original set of explanatory variables.
     
\subsection{Analysis of GSE118657 Data Set}
Gene Expression Omnibus dataset (GSE118657) is a Phase II/III randomized controlled trial examining the use of lactoferrin to prevent nosocomial infections in critically ill patients undergoing mechanical ventilation \citep{Lee2019,Muscedere2018}. This data set consists of 61 patients, among which 32 patients were randomized to receive lactoferrin and the remaining ones were assigned to the placebo group. We are interested in studying the effect of lactoferrin on the length of stay in ICU. For covariate adjustment, we consider four important variables of patients before receiving the treatment-age, sex, SOFA score, and APACHE II score, denoted by $X_b$, and gene expression data of patients, denoted by $X_g$. In the following data analysis, approaches of \citet{Zhang2018} and \citet{Tsiatis2008} are based on modelling $\mathrm{E}[Y|X_b,D=d],d=0,1$ with Forward-1 or Forward-2 model. To make use of information of the gene expression data, we model $\mathrm{E}[Y|X_b,X_g,D=d],d=0,1$ by ML methods and subsequently apply the approach of \citet{Wager2016} and our proposed MDEL approach. Since the dimension of $X=(X_b,X_g)$, $p\approx 50000$, is too high, we use sure independent screening (SIS) method \citep{Fan2008} to filter out variables that are relatively weak-correlated with the response, and reduce the dimension of $X$ to a low level, say $d_X=O(n)$, before modelling $\mathrm{E}[Y|X,D=1]$ and $\mathrm{E}[Y|X,D=0]$.\par

\begin{table}[H]
     	\centering
     	\caption{Point and interval estimates of $\theta$ for GSE118657 data with 5 folds.}
     	\resizebox{\linewidth}{!}{
     		\begin{threeparttable}
		\begin{tabular}{lllllllllll}
			\toprule
			Estimator &       & \multicolumn{1}{l}{Estimate} &       & \multicolumn{1}{l}{SE} &       & \multicolumn{1}{l}{Relative Efficiency} &       & 95\% Confidence Interval &       & 99\% Confidence Interval \\
			\midrule
			\multicolumn{11}{c}{$d_X=100$} \\
			$\widehat{\theta}_{\text{dim}}$ &       & -8.489 &       & 13.737 &       & 1.000 &       & ( -35.41 , 18.44 ) &       & ( -43.87 , 26.90 ) \\
			$\widehat{\theta}_{\text{tdzl}}$(Forward-1) &       & -7.769 &       & 13.701 &       & 1.005 &       & ( -34.62 , 19.08 ) &       & ( -43.06 , 27.52 ) \\
			$\widehat{\theta}_{\text{tdzl}}$(Forward-2) &       & -10.993 &       & 13.965 &       & 0.968 &       & ( -38.36 , 16.38 ) &       & ( -46.96 , 24.98 ) \\
			$\widehat{\theta}_{\text{Zhang}}$(Forward-1) &       & -7.933 &       & 13.204 &       & 1.082 &       & ( -33.81 , 17.95 ) &       & ( -41.94 , 26.08 ) \\
			$\widehat{\theta}_{\text{Zhang}}$(Forward-2) &       & -10.083 &       & 14.413 &       & 0.908 &       & ( -38.33 , 18.17 ) &       & ( -47.21 , 27.04 ) \\
			$\widehat{\theta}_{\text{wdtt}}$(LASSO) &       & -9.310 &       & 13.789 &       & 0.993 &       & ( -36.34 , 17.72 ) &       & ( -44.83 , 26.21 ) \\
			$\widehat{\theta}_{\text{wdtt}}$(SCAD) &       & -10.493 &       & 14.165 &       & 0.941 &       & ( -38.26 , 17.27 ) &       & ( -46.98 , 25.99 ) \\
			$\widehat{\theta}_{\text{wdtt}}$(RF) &       & -13.732 &       & 14.065 &       & 0.954 &       & ( -41.30 , 13.83 ) &       & ( -49.96 , 22.50 ) \\
			$\widehat{\theta}_{\text{mdel}}$(LASSO) &       & -8.540 &       & 12.970 &       & 1.122 &       & ( -33.96 , 16.88 ) &       & ( -41.95 , 24.87 ) \\
			$\widehat{\theta}_{\text{mdel}}$(SCAD) &       & -8.661 &       & 12.181 &       & 1.272 &       & ( -32.53 , 15.21 ) &       & ( -40.04 , 22.71 ) \\
			$\widehat{\theta}_{\text{mdel}}$(RF) &       & -8.647 &       & 13.457 &       & 1.042 &       & ( -35.02 , 17.73 ) &       & ( -43.31 , 26.02 ) \\
			$\widehat{\theta}_{\text{mdel}}$(MULTI) &       & -7.237 &       & 10.755 &       & 1.632 &       & ( -28.32 , 13.84 ) &       & ( -34.94 , 20.46 ) \\
			\midrule
			\multicolumn{11}{c}{$d_X=500$} \\
			$\widehat{\theta}_{\text{dim}}$ &       & -8.489 &       & 13.737 &       & 1.000 &       & ( -35.41 , 18.44 ) &       & ( -43.87 , 26.90 ) \\
			$\widehat{\theta}_{\text{tdzl}}$(Forward-1) &       & -7.769 &       & 13.701 &       & 1.005 &       & ( -34.62 , 19.08 ) &       & ( -43.06 , 27.52 ) \\
			$\widehat{\theta}_{\text{tdzl}}$(Forward-2) &       & -10.993 &       & 13.965 &       & 0.968 &       & ( -38.36 , 16.38 ) &       & ( -46.96 , 24.98 ) \\
			$\widehat{\theta}_{\text{Zhang}}$(Forward-1) &       & -7.933 &       & 13.204 &       & 1.082 &       & ( -33.81 , 17.95 ) &       & ( -41.94 , 26.08 ) \\
			$\widehat{\theta}_{\text{Zhang}}$(Forward-2) &       & -10.083 &       & 14.413 &       & 0.908 &       & ( -38.33 , 18.17 ) &       & ( -47.21 , 27.04 ) \\
			$\widehat{\theta}_{\text{wdtt}}$(LASSO) &       & -10.720 &       & 14.094 &       & 0.950 &       & ( -38.34 , 16.90 ) &       & ( -47.02 , 25.58 ) \\
			$\widehat{\theta}_{\text{wdtt}}$(SCAD) &       & -11.145 &       & 14.072 &       & 0.953 &       & ( -38.73 , 16.44 ) &       & ( -47.39 , 25.10 ) \\
			$\widehat{\theta}_{\text{wdtt}}$(RF) &       & -15.856 &       & 14.167 &       & 0.940 &       & ( -43.62 , 11.91 ) &       & ( -52.35 , 20.64 ) \\
			$\widehat{\theta}_{\text{mdel}}$(LASSO) &       & -8.193 &       & 12.881 &       & 1.137 &       & ( -33.44 , 17.05 ) &       & ( -41.37 , 24.99 ) \\
			$\widehat{\theta}_{\text{mdel}}$(SCAD) &       & -6.061 &       & 12.325 &       & 1.242 &       & ( -30.22 , 18.10 ) &       & ( -37.81 , 25.69 ) \\
			$\widehat{\theta}_{\text{mdel}}$(RF) &       & -8.648 &       & 13.420 &       & 1.048 &       & ( -34.95 , 17.66 ) &       & ( -43.22 , 25.92 ) \\
			$\widehat{\theta}_{\text{mdel}}$(MULTI) &       & -7.202 &       & 11.924 &       & 1.327 &       & ( -30.57 , 16.17 ) &       & ( -37.92 , 23.51 ) \\
			\bottomrule
		\end{tabular}%
	
     		\begin{tablenotes}[flushleft]
     				\small
     				\item SE = standard error, Relative Efficiency = (SE$^2$ of corresponding estimator)/(SE$^2$ of $\widehat{\theta}_{\text{dim}}$).
     			\end{tablenotes}
     		\end{threeparttable}
     	}
     	\label{tab:realdata2}%
     \end{table}%

Results given in Table \ref{tab:realdata2} indicate that there is no improvement about the length of stay in ICU for patients after the use of lactoferrin. Our approach with multiple ML methods, $\widehat{\theta}_{\text{mdel}}$(MULTI), is more efficient than other estimators with the shortest confidence intervals.

\section{Conclusions and Further Discussions}
     In this paper, we propose a machine learning and data-splitting based EL approach to make statistical inference on the average treatment effect in randomized controlled trials. Our approach not only maintains the advantages of the traditional EL approaches, but also overcomes the shortage that the traditional EL approaches usually make invalid inference in high-dimensional settings. Compared with the regression adjustment approach proposed by \citet{Wager2016}, our proposed approach has two attractive characteristics, which are illustrated by our simulation studies: (i). Compared with semiparametric estimators, our proposed estimators perform better when we use random forests to estimate the nuisance parameters; (ii). Our MDEL estimators with multiple ML models are likely to perform as good as that with the oracle model, known as multiple robustness.\par 
     For future work, we plan to (i) study the asymptotic theory of the proposed EL estimator with multiple models; (ii) generalize our proposed approach to high-dimensional observational studies by modelling propensity scores and imposing additional constraints about the propensity scores.

\section*{Acknowledgements}
We are grateful to two reviewers' helpful and instructive comments and suggestions that have greatly improved our original submission.
\section*{Funding}
Ying Yan's research is supported by the National Natural Science Foundation of China (NSFC) (Grant No. 11901599; Grant  No. 71991474).

\bibliographystyle{asa}
\bibliography{mybib}

\newpage
\section*{Appendix: Lemmas and Proofs}

\subsection*{Lemmas}
     \begin{lemma}[\citet{Chernozhukov2018}; Conditional Convergence implies Unconditional]\label{lemma1}
     Let $\{X_m\}$ and $\{Y_m\}$ be random vectors. (a) If for $\varepsilon_m\rightarrow0$, $\mathbb{P}\left(||X_m||>\varepsilon_m|Y_m\right)\rightarrow0$ as $m\rightarrow\infty$, then $\mathbb{P}(||X_m||>\varepsilon_m)\rightarrow0$ as $m\rightarrow\infty$. (b) Let $\{A_m\}$ be a sequence of positive constants. If $||X_m||=O_p(A_m)$ conditional on $Y_m$, then $||X_m||=O_p(A_m)$ holds unconditionally.
     \end{lemma}
     \begin{lemma}\label{lemma3}
     $\widehat{\textsc{G}}\left(X_i,\widehat{\textsl{g}}^{(d)}_{k},\widehat{\xi}^{(d)}\right)=\widehat{\textsc{G}}\left(X_i,\widehat{\textsl{g}}^{(d)}_{k},\ddot{\xi}^{(d)}\right)+O_p(\frac{1}{\sqrt{n}})$, $i\in\mathbb{I}^{(d)}_k$ for $k=1,\cdots,K$ and $d=0,1$.
     \end{lemma}
     \begin{proof}
     Simple algebra gives 
       \[
       \widehat{\textsc{G}}\left(X_i,\widehat{\textsl{g}}^{(d)}_{k},\ddot{\xi}^{(d)}\right)-\widehat{\textsc{G}}\left(X_i,\widehat{\textsl{g}}^{(d)}_{k},\widehat{\xi}^{(d)}\right)=\frac{1}{K}\sum\limits_{k=1}^K\frac{1}{|\mathbb{I}_k|}\sum\limits_{i\in \mathbb{I}_k}\left(\widehat{\textsl{g}}_k^{(d)}(X_i)-\mathrm{E}\left[\left.\widehat{\textsl{g}}_k^{(d)}(X_i)\right\vert (W_j)_{j\in \mathbb{I}_k^{(d)^c}}\right]\right).
       \]
       The proof is completed by the Central Limit Theorem and lemma \ref{lemma1}.
       \end{proof}
       
     \subsection*{Proofs}
     \textbf{Proof of Proposition \ref{pro:1}}
     \begin{proof}
     Let $n_d^*=\frac{n_d}{K}$. $n_d^*$ is an integer as $\mathbb{I}_k^{(d)}$, $k=1,\cdots,K$ are of equal size for $d=0,1$. Write $(W_i)_{i\in\mathbb{I}_k^{(d)}}=\{W^d_{k1},W^d_{k2},\cdots,W^d_{kn_d^*}\}$ for $k=1,\cdots,K$ and $d=0,1$ with random orders. Let\\ $T_{udj}=\frac{1}{K}\sum\limits_{k=1}^K\left[\widehat{\textsc{G}}\left(X^d_{kj},\widehat{\textsl{g}}^{(d)}_{k},\widehat{\xi}^{(d)}\right)\right]_u$ for $u=1,\cdots,r$, $j=1,\cdots,n_d^*$ and $d=0,1$. It suffices to prove that $0$ is contained in the convex hull of $\{T_{ud1},\cdots,T_{udn_d^*}\}$ with probability tending to $1$ as $n\rightarrow\infty$ for $u=1,\cdots,r$ and $d=0,1$. To prove it, it suffices to show that for any given $u$ and $d$, $\mathrm{P}(\max\limits_{1\leq j\leq n_d^*}T_{udj}\leq 0)\rightarrow 0$ and $\mathrm{P}(\min\limits_{1\leq j\leq n_d^*}T_{udj}\geq 0)\rightarrow 0$. Now, we only prove that $\mathrm{P}(\max\limits_{1\leq j\leq n_d^*}T_{udj}\leq 0)\rightarrow 0$ and the proof of $\mathrm{P}(\min\limits_{1\leq j\leq n_d^*}T_{udj}\geq 0)\rightarrow 0$ will be similar. Let $T^{q}_{udj}=\left[\widehat{\textsl{g}}^{(d)}_{q}(X^d_{qj})\right]_u-\left[\widehat{\xi}_q^{(d)}\right]_u$ for $q=1,\cdots,K$. It is easy to check that $T_{udj}=\frac{1}{K}\sum\limits_{q=1}^KT^{q}_{udj}$. Simple calculation gives
     \[
        \mathrm{P}(\max\limits_{1\leq j\leq n_d^*}T_{udj}\leq 0)=\mathrm{P}(\max\limits_{1\leq j\leq n_d^*}\frac{1}{K}\sum\limits_{q=1}^KT^{q}_{udj}\leq 0)\leq \mathrm{P}(\min\limits_{1\leq q\leq K}\max\limits_{1\leq j\leq n_d^*}T^{q}_{udj}\leq 0).
     \]
     Therefore, it suffices to prove that $\mathrm{P}(\max\limits_{1\leq j\leq n_d^*}T^{q}_{udj}\leq 0)\rightarrow0$ for $q=1,\cdots,K$. The proof below follows the technique of \citet{Jing2009}. For a given $q$, let $\nu_{udj}=\psi(T^{q}_{udj})$, where $\psi(x)$ is a nondecreasing, twice differentiable function with bounded first and second derivatives such that
     \begin{equation}
      \psi(x)=\left\{\begin{array}{ll}
      0, & \text { if } x \leq 0 \\
        a(x), & \text { if } 0<x<\epsilon \\
         1, & \text { if } x \geq \epsilon
         \end{array}\right.
    \end{equation}
    with $0<a(x)<1$ for $0<x<\epsilon$. Simple algebra gives that, conditional on $(W_i)_{i\in\mathbb{I}_q^{(d)^c}}$,
    \begin{align*}
     \mathrm{P}\left(\max\limits_{1\leq j\leq n_d^*}T^{q}_{udj}\leq 0\right)=&\mathrm{P}\left(\nu_{ud1}=0,\cdots,\nu_{udn_d^*}=0\right)\\
     =&\mathrm{P}\left(\sum\limits_{j=1}^{n_d^*}\nu_{udj}=0\right)\\
     \leq& \mathrm{P}\left(\left|\sum\limits_{j=1}^{n_d^*}(\nu_{udj}-\mathrm{E}\nu_{udj})\right|\geq n_d^*\mathrm{E}\nu_{udj}\right)\\
     \leq& \frac{\mathrm{E}\left[\left(\sum\limits_{j=1}^{n_d^*}(\nu_{udj}-\mathrm{E}\nu_{udj})\right)^2\right]}{n^2(\mathrm{E}\nu_{udj})^2}\\
     \leq& \frac{n_d^*\mathrm{Var}(\nu_{ud1})+n_d^*(n_d^*-1)\mathrm{Cov}(\nu_{ud1},\nu_{ud2})}{{n_d^*}^2(\mathrm{E}\nu_{ud1})^2}
    \end{align*}
    It suffices to show that, conditional on $(W_i)_{i\in\mathbb{I}_q^{(d)^c}}$,
    \begin{itemize}
        \item[(a)] $\mathrm{Var}(\nu_{ud1})\leq 1$.
        \item[(b)] $\lim\limits_{n\rightarrow\infty}\mathrm{E}\nu_{ud1}\geq c>0$.
        \item[(c)] $\mathrm{Cov}(\nu_{ud1},\nu_{ud2})\rightarrow 0$.
    \end{itemize}
    (a) holds because $0\leq \nu_{ud1}\leq 1$. Simple algebra indicates that 
    \begin{align*}
    T^{q}_{udj}=&\left[\widehat{\textsl{g}}^{(d)}_{q}(X^d_{qj})\right]_u-\mathrm{E}\left[\left.\left[\widehat{\textsl{g}}^{(d)}_{q}(X^d_{qj})\right]_u\right\vert(W_i)_{i\in\mathbb{I}_q^{(d)^c}}\right]-\left[\widehat{\xi}_q^{(d)}\right]_u+\mathrm{E}\left[\left.\left[\widehat{\textsl{g}}^{(d)}_{q}(X^d_{qj})\right]_u\right\vert(W_i)_{i\in\mathbb{I}_q^{(d)^c}}\right]\\
    =& Q_u(X_{qj}^d,\widehat{\textsl{g}}^{(d)}_{q})+\mathrm{Rem}^q_{ud}.
    \end{align*}
    Under condition (\textbf{A1}), The Central Limit Theory (CLT) and Lemma \ref{lemma1} gives $\mathrm{Rem}^q_{ud}=O_p(\frac{1}{\sqrt{n}})$. A Taylor expansion yields
    \[
    \nu_{udj}=\psi(T^{q}_{udj})=\psi(Q_u(X_{qj}^d,\widehat{\textsl{g}}^{(d)}_{q}))+C_n\mathrm{Rem}^q_{ud},
    \]
    where $C_n<C$ for some constant $C$. Therefore, it is easy to verify that $\mathrm{E}\nu_{udj}\rightarrow \mathrm{E}\left[\psi(Q_u(X_{qj}^d,\widehat{\textsl{g}}^{(d)}_{q}))\right]$. Note that conditional on $(W_i)_{i\in\mathbb{I}_q^{(d)^c}}$, $\mathrm{E}\left[Q_u(X_{qj}^d,\widehat{\textsl{g}}^{(d)}_{q})\right]=0$ and $\mathrm{Var}\left(Q_u(X_{qj}^d,\widehat{\textsl{g}}^{(d)}_{q})\right)>0$ (By condition (\textbf{A1})). Therefore, we have $\mathrm{P}(Q_u(X_{qj}^d,\widehat{\textsl{g}}^{(d)}_{q})>0)>0$, which implies that $\mathrm{E}\left[\psi(Q_u(X_{qj}^d,\widehat{\textsl{g}}^{(d)}_{q}))\right]>0$ conditional on $(W_i)_{i\in\mathbb{I}_q^{(d)^c}}$. This completes the proof of (b). (c) is obvious because conditional on $(W_i)_{i\in\mathbb{I}_q^{(d)^c}}$, $\psi(Q_u(X_{q1}^d,\widehat{\textsl{g}}^{(d)}_{q}))$ and $\psi(Q_u(X_{q2}^d,\widehat{\textsl{g}}^{(d)}_{q}))$ are independent and $C_n\mathrm{Rem}^q_{ud}\rightarrow 0$ with probability tending to 1. This completes the proof.
     \end{proof}
     \noindent
     \textbf{Proof of Proposition \ref{pro:2}}
     \begin{proof}
     	Without loss of generality, we assume that $q_1=q_2=r$. Otherwise we can replace\\ $\widehat{\textsc{G}}\left(X_i,\widehat{\textsl{g}}^{(d)}_{k},\widehat{\xi}^{(d)}\right)$ by a subset of $q_1$ components that guarantee the limit of $\widehat{V}_n^{(d)}$ is bounded by two finite and positive definite matrices of the full rank $q_1$ (See Chapter 11 in \citet{Owen2001}). For fixed $d$, from $\frac{1}{n_d}\sum\limits_{k=1}^K\sum\limits_{i\in \mathbb{I}_k^{(d)}}\frac{\widehat{\textsc{G}}\left(X_i,\widehat{\textsl{g}}^{(d)}_{k},\widehat{\xi}^{(d)}\right)}{1+\widehat{\lambda}_d^{\tau} \widehat{\textsc{G}}\left(X_i,\widehat{\textsl{g}}^{(d)}_{k},\widehat{\xi}^{(d)}\right)}=0$, simple algebra gives
     	\begin{align*}
     	0&=\frac{1}{n_d}\sum\limits_{k=1}^K\sum\limits_{i\in \mathbb{I}_k^{(d)}}\widehat{\textsc{G}}\left(X_i,\widehat{\textsl{g}}^{(d)}_{k},\widehat{\xi}^{(d)}\right)\left[1-\frac{1}{1+\widehat{\lambda}_d^{\tau} \widehat{\textsc{G}}\left(X_i,\widehat{\textsl{g}}^{(d)}_{k},\widehat{\xi}^{(d)}\right)}\right]\\
     	&=\frac{1}{n_d}\sum\limits_{k=1}^K\sum\limits_{i\in \mathbb{I}_k^{(d)}}\widehat{\textsc{G}}\left(X_i,\widehat{\textsl{g}}^{(d)}_{k},\widehat{\xi}^{(d)}\right)-\frac{1}{n_d}\sum\limits_{k=1}^K\sum\limits_{i\in \mathbb{I}_k^{(d)}}\frac{\widehat{\textsc{G}}\left(X_i,\widehat{\textsl{g}}^{(d)}_{k},\widehat{\xi}^{(d)}\right)^{\otimes2}\widehat{\lambda}_d}{1+\widehat{\lambda}_d^\tau \widehat{\textsc{G}}\left(X_i,\widehat{\textsl{g}}^{(d)}_{k},\widehat{\xi}^{(d)}\right)}.
     	\end{align*}
     	Therefore, we have
     	\begin{equation}\label{eq:1}
     	\frac{1}{n_d}\sum\limits_{k=1}^K\sum\limits_{i\in \mathbb{I}_k^{(d)}}\widehat{\textsc{G}}\left(X_i,\widehat{\textsl{g}}^{(d)}_{k},\widehat{\xi}^{(d)}\right)=\frac{1}{n_d}\sum\limits_{k=1}^K\sum\limits_{i\in \mathbb{I}_k^{(d)}}\frac{\widehat{\textsc{G}}\left(X_i,\widehat{\textsl{g}}^{(d)}_{k},\widehat{\xi}^{(d)}\right)^{\otimes2}\widehat{\lambda}_d}{1+\widehat{\lambda}_d^\tau \widehat{\textsc{G}}\left(X_i,\widehat{\textsl{g}}^{(d)}_{k},\widehat{\xi}^{(d)}\right)}.
     	\end{equation}
     	For fixed $k$, conditional on $\left(W_i\right)_{i\in \mathbb{I}^{(d)^c}_k}$, the Central Limit Theorem indicates that
     	\[
     	\frac{1}{|\mathbb{I}_k^{(d)}|}\sum\limits_{i\in \mathbb{I}_k^{(d)}}\left(\widehat{\textsl{g}}^{(d)}_k(X_i)-\widehat{\xi}^{(d)}_k\right)-\mathrm{E}\left[\left.\widehat{\textsl{g}}^{(d)}_k(X)-\widehat{\xi}^{(d)}_k\right\vert \left(W_i\right)_{i\in \mathbb{I}^{(d)^c}_k}\right]=O_p(\frac{1}{\sqrt{n}}).
     	\]
        Then, lemma \ref{lemma1} gives $\frac{1}{|\mathbb{I}_k^{(d)}|}\sum\limits_{i\in \mathbb{I}_k^{(d)}}\left(\widehat{\textsl{g}}^{(d)}_k(X_i)-\widehat{\xi}^{(d)}_k\right)=O_p(\frac{1}{\sqrt{n}})$ unconditionally. Therefore, the left term of (\ref{eq:1}) is 
        \[
        \frac{1}{n_d}\sum\limits_{k=1}^K\sum\limits_{i\in \mathbb{I}_k^{(d)}}\widehat{\textsc{G}}\left(X_i,\widehat{\textsl{g}}^{(d)}_{k},\widehat{\xi}^{(d)}\right)=\frac{1}{K}\sum\limits_{k=1}^K\frac{1}{|\mathbb{I}_k^{(d)}|}\sum\limits_{i\in \mathbb{I}_k^{(d)}}\left(\widehat{\textsl{g}}^{(d)}_k(X_i)-\widehat{\xi}^{(d)}_k\right)=O_p(\frac{1}{\sqrt{n}}).
        \]
        Turn to the right term of (\ref{eq:1}), and let $\nu_d=\frac{\widehat{\lambda}_d}{||\widehat{\lambda}_d||}$, where  $||\cdot||$ is the Euclidean norm. We have 
        \[
        1+\widehat{\lambda}_d^\tau \widehat{\textsc{G}}\left(X_i,\widehat{\textsl{g}}^{(d)}_{k},\widehat{\xi}^{(d)}\right)\leq1+||\widehat{\lambda}_d||\nu_d^\tau\widehat{\textsc{G}}\left(X_i,\widehat{\textsl{g}}^{(d)}_{k},\widehat{\xi}^{(d)}\right)\leq 1+2||\widehat{\lambda}_d||\sqrt{r}\max\limits_{k\in\{1,\cdots,K\}}\max\limits_{j=1,\cdots,r }\max\limits_{i\in \mathbb{I}_k}\left\vert[\widehat{\textsl{g}}^{(d)}_k(X_i)]_j\right\vert.
        \]
        Condition (\textbf{A1}), lemma 11.2 in \citet{Owen2001}, and lemma \ref{lemma1} indicate that $\max\limits_{k\in\{1,\cdots,K\}}\max\limits_{j=1,\cdots,r }\max\limits_{i\in \mathbb{I}_k}\left\vert[\widehat{\textsl{g}}^{(d)}_k(X_i)]_j\right\vert=o_p(n^{1/2})$. Multiply $\nu_d^\tau$ on both sides of (\ref{eq:1}), we have
     	\begin{align*}
     	&||\widehat{\lambda}_d||\frac{1}{n_d}\sum\limits_{k=1}^K\sum\limits_{i\in \mathbb{I}_k^{(d)}}\nu_d^\tau\widehat{\textsc{G}}\left(X_i,\widehat{\textsl{g}}^{(d)}_{k},\widehat{\xi}^{(d)}\right)^{\otimes2}\nu_d\\
     	\leq
     	&\frac{\nu_d^\tau}{n_d}\sum\limits_{k=1}^K\sum\limits_{i\in \mathbb{I}_k^{(d)}}\widehat{\textsc{G}}\left(X_i,\widehat{\textsl{g}}^{(d)}_{k},\widehat{\xi}^{(d)}\right)\left(1+2||\widehat{\lambda}_d||\sqrt{r}\max\limits_{k\in\{1,\cdots,K\}}\max\limits_{j=1,\cdots,r }\max\limits_{i\in \mathbb{I}_k}\left\vert[\widehat{\textsl{g}}^{(d)}_k(X_i)]_j\right\vert\right).
     	\end{align*}
     	Condition (\textbf{A2}) implies $\frac{1}{n_d}\sum\limits_{k=1}^K\sum\limits_{i\in \mathbb{I}_k^{(d)}}\nu_d^\tau\widehat{\textsc{G}}\left(X_i,\widehat{\textsl{g}}^{(d)}_{k},\widehat{\xi}^{(d)}\right)^{\otimes2}\nu_d\asymp1$. It follows from all above results that
     	\begin{align}
     	||\widehat{\lambda}_d||\leq O_p(\frac{1}{\sqrt{n}})(1+2||\widehat{\lambda}_d||o_p(n^{1/2})).
     	\label{eq: 2}
     	\end{align}
     	Equation (\ref{eq: 2}) indicates that $||\hat{\lambda}_d||=O_p(\frac{1}{\sqrt{n}})$. This completes the proof.
       \end{proof}
     
     ~
     
   \noindent
   \textbf{Proof of Proposition \ref{pro:3}}
   \begin{proof}
     First, we consider the case $d=1$ and the case $d=0$ will be similar. Taylor expansion, Proposition \ref{pro:1}, and Lemma \ref{lemma3} lead to
     \begin{equation}\label{eq4}
     \begin{aligned}
     0 &= \sqrt{n}\frac{1}{n_1}\sum\limits_{k=1}^K\sum\limits_{i\in \mathbb{I}_k^{(1)}}\frac{\widehat{\textsc{G}}\left(X_i,\widehat{\textsl{g}}^{(1)}_{k},\widehat{\xi}^{(1)}\right)}{1+\widehat{\lambda}_1^{\tau} \widehat{\textsc{G}}\left(X_i,\widehat{\textsl{g}}^{(1)}_{k},\widehat{\xi}^{(1)}\right)}\\
     &=
     \frac{1}{\sqrt{n}}\sum\limits_{k=1}^K\sum\limits_{i\in \mathbb{I}_k}\frac{D_i}{\delta}\widehat{\textsc{G}}\left(X_i,\widehat{\textsl{g}}^{(1)}_{k},\widehat{\xi}^{(1)}\right)-\sqrt{n}\frac{1}{K}\sum\limits_{k=1}^K\frac{1}{|\mathbb{I}_k|}\sum\limits_{i\in \mathbb{I}_k}\frac{D_i}{\delta}\widehat{\textsc{G}}\left(X_i,\widehat{\textsl{g}}^{(1)}_{k},\widehat{\xi}^{(1)}\right)^{\otimes2}\widehat{\lambda}_1+o_p(1)\\
     &=
     \frac{1}{\sqrt{n}}\sum\limits_{k=1}^K\sum\limits_{i\in \mathbb{I}_k}\frac{D_i-\delta}{\delta}\widehat{\textsc{G}}\left(X_i,\widehat{\textsl{g}}^{(1)}_{k},\ddot{\xi}^{(1)}\right)-\sqrt{n}\frac{1}{K}\sum\limits_{k=1}^K\frac{1}{|\mathbb{I}_k|}\sum\limits_{i\in \mathbb{I}_k}\frac{D_i}{\delta}\widehat{\textsc{G}}\left(X_i,\widehat{\textsl{g}}^{(1)}_{k},\ddot{\xi}^{(1)}\right)^{\otimes2}\widehat{\lambda}_1+o_p(1).\\
     \end{aligned}
     \end{equation}
     Therefore, we have
     \begin{equation}\label{eq2}
     \sqrt{n}\widehat{\lambda}_1=\ddot{{S}}_n^{(1)^{-1}}\frac{1}{\sqrt{n}}\sum\limits_{k=1}^K\sum\limits_{i\in \mathbb{I}_k}\frac{D_i-\delta}{\delta}\widehat{\textsc{G}}\left(X_i,\widehat{\textsl{g}}^{(1)}_{k},\ddot{\xi}^{(1)}\right)+o_p(1).
     \end{equation}
     Under condition (\textbf{A3}) and (\textbf{A4}), by Taylor expansion, Proposition \ref{pro:1}, Lemma \ref{lemma3}, and (\ref{eq2}), we have
     \begin{equation}\label{eq3}
     \begin{aligned}
     \sqrt{n}\left(\widehat{\theta}^{(1)}_{\text{mdel}}-\theta_1\right)
     &=
     \sqrt{n}\sum\limits_{k=1}^K\sum\limits_{i\in \mathbb{I}_k}D_i\widehat{p}_i\left(Y_i-\theta_1\right)\\
        &=
        \frac{1}{\sqrt{n}}\sum\limits_{k=1}^K\sum\limits_{i\in \mathbb{I}_k}\frac{D_i}{\delta}\frac{Y_i-\theta_1}{1+\widehat{\lambda}_1^\tau \widehat{\textsc{G}}\left(X_i,\widehat{\textsl{g}}^{(1)}_{k},\widehat{\xi}^{(1)}\right)}+o_p(1)\\
        &=
        \frac{1}{\sqrt{n}}\sum\limits_{k=1}^K\sum\limits_{i\in \mathbb{I}_k}\frac{D_i}{\delta}\left(Y_i-\theta_1\right)-\frac{1}{\sqrt{n}}\sum\limits_{k=1}^K\sum\limits_{i\in \mathbb{I}_k}\frac{D_i}{\delta}\left(Y_i-\theta_1\right)\widehat{\textsc{G}}\left(X_i,\widehat{\textsl{g}}^{(1)}_{k},\ddot{\xi}^{(1)}\right)\widehat{\lambda}_1+o_p(1)\\
        &=
        \frac{1}{\sqrt{n}}\sum\limits_{k=1}^K\sum\limits_{i\in \mathbb{I}_k}\frac{D_i}{\delta}\left(Y_i-\theta_1\right)-\ddot{J}_n^{(1)^\tau}\ddot{{S}}_n^{(1)^{-1}}\frac{1}{\sqrt{n}}\sum\limits_{k=1}^K\sum\limits_{i\in \mathbb{I}_k}\frac{D_i-\delta}{\delta}\widehat{\textsc{G}}\left(X_i,\widehat{\textsl{g}}^{(1)}_{k},\ddot{\xi}^{(1)}\right)+o_p(1)\\
        &=
        \frac{1}{\sqrt{n}}\sum\limits_{k=1}^K\sum\limits_{i\in \mathbb{I}_k}\left[\frac{D_i}{\delta}\left(Y_i-\theta_1\right)-\frac{D_i-\delta}{\delta}\ddot{J}_n^{(1)^\tau}\ddot{{S}}_n^{(1)^{-1}}\widehat{\textsc{G}}\left(X_i,\widehat{\textsl{g}}^{(1)}_{k},\ddot{\xi}^{(1)}\right)\right]+o_p(1)
        \end{aligned}
     	\end{equation}
     	It is easy to give the form of $\sqrt{n}\left(\widehat{\theta}^{(0)}_{\text{mdel}}-\theta_0\right)$ in a similar way:
     	\[
     	\sqrt{n}\left(\widehat{\theta}^{(0)}_{\text{mdel}}-\theta_0\right)=\frac{1}{\sqrt{n}}\sum\limits_{k=1}^K\sum\limits_{i\in \mathbb{I}_k}\left[\frac{1-D_i}{1-\delta}\left(Y_i-\theta_0\right)-\frac{D_i-\delta}{1-\delta}\ddot{J}_n^{(0)^\tau}\ddot{{S}}_n^{(0)^{-1}}\widehat{\textsc{G}}\left(X_i,\widehat{\textsl{g}}^{(0)}_{k},\ddot{\xi}^{(0)}\right)\right]+o_p(1).
     	\]
     	Above results give that
     	\begin{equation*}
     	\begin{aligned}
     		\sqrt{n}\left(\widehat{\theta}_{\text{mdel}}-\theta\right)=\frac{1}{\sqrt{n}}\sum\limits_{k=1}^K\sum\limits_{i\in \mathbb{I}_k}\left[\frac{D_i}{\delta}\left(Y_i-\theta_1\right)-\frac{D_i-\delta}{\delta}\ddot{J}_n^{(1)^\tau}\ddot{{S}}_n^{(1)^{-1}}\widehat{\textsc{G}}\left(X_i,\widehat{\textsl{g}}^{(1)}_{k},\ddot{\xi}^{(1)}\right)\right.\\
     		\left.-\frac{1-D_i}{1-\delta}\left(Y_i-\theta_0\right)+\frac{D_i-\delta}{1-\delta}\ddot{J}_n^{(0)^\tau}\ddot{{S}}_n^{(0)^{-1}}\widehat{\textsc{G}}\left(X_i,\widehat{\textsl{g}}^{(0)}_{k},\ddot{\xi}^{(0)}\right)\right]+o_p(1).
     	\end{aligned}
     	\end{equation*}
     	This completes the proof.
     \end{proof}
     	
   ~
   
   \noindent
   \textbf{Proof of Theorem \ref{thm 1}}
     \begin{proof}
     		When $r=1$, conditional on $\left(W_i\right)_{i\in \mathbb{I}_k^{(d)^c}}$, the Holder Inequality gives that
     		\[
     		\mathrm{E}\left[\left.\left|\widehat{\textsl{g}}^{(d)}_k(X)-\eta^{(d)}(X)\right|\quad\right\vert \left(W_i\right)_{i\in \mathbb{I}_k^{(d)^c}}\right]\leq \sqrt{\mathrm{E}\left[\left.\left(\widehat{\textsl{g}}^{(d)}_k(X)-\eta^{(d)}(X)\right)^2\right\vert \left(W_i\right)_{i\in \mathbb{I}_k^{(d)^c}}\right]}.
     		\]
     		Therefore, we have $\mathrm{E}\left[\left.\left|\widehat{\textsl{g}}^{(d)}_k(X)-\eta^{(d)}(X)\right|\quad\right\vert \left(W_i\right)_{i\in \mathbb{I}_k^{(d)^c}}\right]\rightarrow0$ in probability as $n\rightarrow\infty$ for $k=1,\cdots,K$. Let $\textsc{G}(X_i,\eta^{(d)},\theta_d)=\eta^{(d)}(X_i)-\theta_d$. For simplicity, write $\varsigma_k^{(d)}(X_i)=\widehat{\textsl{g}}^{(d)}_k(X_i)-\eta^{(d)}(X_i)$. It is straightforward to show that $\widehat{\textsc{G}}\left(X_i,\widehat{\textsl{g}}^{(d)}_{k},\ddot{\xi}^{(d)}\right)-\textsc{G}(X_i,\eta^{(d)},\theta_d)=\varsigma_k^{(d)}(X_i)+o_p(1)$ by lemma \ref{lemma1}. Following (\ref{eq4}), it is easy to verify that
     		\begin{equation}
     		\begin{aligned}
     		0&=\frac{1}{\sqrt{n}}\sum\limits_{i\in I}\frac{D_i-\delta}{\delta}\textsc{G}(X_i,\eta^{(1)},\theta_1)-\sqrt{n}\frac{1}{n}\sum\limits_{i\in I}\frac{D_i}{\delta}\textsc{G}(X_i,\eta^{(1)},\theta_1)^{2}\widehat{\lambda}_1+o_p(1)\\
     		&+\frac{1}{\sqrt{n}}\sum\limits_{k=1}^K\sum\limits_{i\in \mathbb{I}_k}\frac{D_i-\delta}{\delta}\varsigma_k^{(1)}(X_i)+\frac{1}{\sqrt{n}}\sum\limits_{k=1}^K\sum\limits_{i\in \mathbb{I}_k}\varsigma_k^{(1)}(X_i)^2\widehat{\lambda}_1-\frac{2}{\sqrt{n}}\sum\limits_{k=1}^K\sum\limits_{i\in \mathbb{I}_k}\frac{D_i}{\delta}\varsigma_k^{(1)}(X_i)\textsc{G}(X_i,\eta^{(1)},\theta_1)\widehat{\lambda}_1.
     		\end{aligned}
     		\end{equation}
     		Now, we bound 
     		\[
     		A=\frac{1}{\sqrt{n}}\sum\limits_{k=1}^K\sum\limits_{i\in \mathbb{I}_k}\frac{D_i-\delta}{\delta}\varsigma_k^{(1)}(X_i),\quad B=\frac{1}{\sqrt{n}}\sum\limits_{k=1}^K\sum\limits_{i\in \mathbb{I}_k}\varsigma_k^{(1)}(X_i)^2\widehat{\lambda}_1
     		\]
     		and
     		\[
     		C=\frac{1}{\sqrt{n}}\sum\limits_{k=1}^K\sum\limits_{i\in \mathbb{I}_k}\frac{D_i}{\delta}\varsigma_k^{(1)}(X_i)\textsc{G}(X_i,\eta^{(1)},\theta_1)\widehat{\lambda}_1,
     		\]
     		respectively. Conditional on $(W_i)_{i\in \mathbb{I}^{(1)^c}_k}$, the mean of $\frac{1}{\sqrt{|\mathbb{I}_k|}}\sum\limits_{i\in \mathbb{I}_k}\frac{D_i-\delta}{\delta}\varsigma_k^{(1)}(X_i)$ is zero and the variance is given by
     		\[
     	    \mathrm{E}[(D-\delta)^2]\cdot\mathrm{E}\left[\left.\varsigma_k^{(1)}(X)^2\right\vert(W_i)_{i\in \mathbb{I}^{(1)^c}_k}\right],
     		\]
     		which converges to zero in probability as $n\rightarrow\infty$. Then $A=o_p(1)$ by the Chebyshev's Inequality and lemma \ref{lemma1}. $B$ vanishes in probability because $\sqrt{n}\widehat{\lambda}_1=O_p(1)$.  For $C$, the Cauchy-Schwarz Inequality gives that
     		\[
     		C\leq \sqrt{n}\widehat{\lambda}_1\frac{1}{K}\sum\limits_{k=1}^K\sqrt{\frac{1}{|\mathbb{I}_k|}\sum\limits_{i\in \mathbb{I}_k}\varsigma_k^{(1)}(X_i)^2}\cdot\sqrt{\frac{1}{|\mathbb{I}_k|}\sum\limits_{i\in \mathbb{I}_k}\left(\frac{D_i}{\delta}\textsc{G}(X_i,\eta^{(1)},\theta_1)\right)^2}.
     		\]
     		Conditional on $(W_i)_{i\in \mathbb{I}_k^{(d)^c}}$, the right term of above inequality converges to $0$ in probability as $n\rightarrow\infty$; therefore $C=o_p(1)$ by lemma \ref{lemma1}.
     		Above results give that 
     		\[
     		\sqrt{n}\widehat{\lambda}_1=\mathrm{E}\left[\textsc{G}(X_i,\eta^{(1)},\theta_1)^2\right]^{-1}\frac{1}{\sqrt{n}}\sum\limits_{i=1}^n\frac{D_i-\delta}{\delta}\textsc{G}(X_i,\eta^{(1)},\theta_1)+o_p(1).
     		\]
     		Similarly, it is easy to check that
     		\[
     		\sqrt{n}\widehat{\lambda}_0=\mathrm{E}\left[\textsc{G}(X_i,\eta^{(0)},\theta_0)^2\right]^{-1}\frac{1}{\sqrt{n}}\sum\limits_{i=1}^n\frac{D_i-\delta}{1-\delta}\textsc{G}(X_i,\eta^{(0)},\theta_0)+o_p(1).
     		\]
     		Using above results, Taylor expansion indicates that
     		\begin{equation*}
     		\begin{aligned}
     		\sqrt{n}\left(\widehat{\theta}^{(1)}_{\text{mdel}}-\theta_1\right) &= \frac{1}{\sqrt{n}}\sum\limits_{i=1}^n\frac{D_i}{\delta}\left(Y_i-\theta_1\right)-\frac{1}{\sqrt{n}}\sum\limits_{i=1}^n\frac{D_i}{\delta}\left(Y_i-\theta_1\right)\textsc{G}(X_i,\eta^{(1)},\theta_1)\widehat{\lambda}_1+o_p(1)\\
     		&=\frac{1}{\sqrt{n}}\sum\limits_{i=1}^n\left\{\frac{D_i}{\delta}\left(Y_i-\theta_1\right)\right.\\
     		&\left.-\frac{D_i-\delta}{\delta}\mathrm{E}\left[\frac{D}{\delta}(Y-\theta_1)\textsc{G}(X,\eta^{(1)},\theta_1)\right]\mathrm{E}\left[\textsc{G}(X,\eta^{(1)},\theta_1)^2\right]^{-1}\textsc{G}(X,\eta^{(1)},\theta_1)\right\}\\
     		&+o_p(1).
     		\end{aligned}
     		\end{equation*}
     		Following from
     		\begin{align*}
     		\mathrm{E}\left[\frac{D}{\delta}(Y-\theta_1)\textsc{G}(X,\eta^{(1)},\theta_1)\right]=\frac{\mathbb{P}(D=1)}{\delta}\mathrm{E}\left[(Y-\theta_1)\textsc{G}(X,\eta^{(1)},\theta_1)|D=1\right]\\
     		=\mathrm{E}\left[\mathrm{E}\left[(Y-\theta_1)|X,D=1\right]\textsc{G}(X,\eta^{(1)},\theta_1)\right]=\mathrm{E}\left[\textsc{G}(X,\eta^{(1)},\theta_1)^2\right],
     		\end{align*}
     		we have
     		\begin{equation*}
     		\begin{aligned}
     		\sqrt{n}\left(\widehat{\theta}^{(1)}_{\text{mdel}}-\theta_1\right)&=
     		\frac{1}{\sqrt{n}}\sum\limits_{i=1}^n\left\{\frac{D_i}{\delta}\left(Y_i-\theta_1\right)-\frac{D_i-\delta}{\delta}\left(\eta^{(1)}(X_i)-\theta_1\right)\right\}+o_p(1)\\
     		&=
     		\frac{1}{\sqrt{n}}\sum\limits_{i=1}^n\left\{\frac{D_i}{\delta}\left(Y_i-\eta^{(1)}(X_i)\right)+\left(\eta^{(1)}(X_i)-\theta_1\right)\right\}+o_p(1).
     		\end{aligned}
     		\end{equation*}
     		Similarly, when $d=0$, it is easy to verify that
     		\begin{equation*}
     		\begin{aligned}
     		\sqrt{n}\left(\widehat{\theta}^{(0)}_{\text{mdel}}-\theta_0\right)
     		=
     		\frac{1}{\sqrt{n}}\sum\limits_{i=1}^n\left\{\frac{1-D_i}{1-\delta}\left(Y_i-\eta^{(0)}(X_i)\right)+\left(\eta^{(0)}(X_i)-\theta_0\right)\right\}+o_p(1).
     		\end{aligned}
     		\end{equation*}
     		Therefore, we have
     		\begin{equation*}
     		\begin{aligned}
     		\sqrt{n}\left(\widehat{\theta}_{\text{mdel}}-\theta\right)
     		&=
     		\frac{1}{\sqrt{n}}\sum\limits_{i=1}^n\left\{\frac{D_i}{\delta}\left(Y_i-\eta^{(1)}(X_i)\right)-\frac{1-D_i}{1-\delta}\left(Y_i-\eta^{(0)}(X_i)\right)\right.\\
     		&\left.+\eta^{(1)}(X_i)-\eta^{(0)}(X_i)-\theta\right\}+o_p(1).
     		\end{aligned}
     		\end{equation*}
     		This completes the proof.
     	    \end{proof}
        
        ~
        
        \noindent
\textbf{Proof of Theorem \ref{thm 2}}
        \begin{proof}
        Following the proof of Theorem \ref{thm 1}, it is easy to verify that
        \[
        \widehat{J}_n^{(d)}=\frac{1}{n_d}\sum\limits_{k=1}^K\sum\limits_{i\in \mathbb{I}^{(d)}_k}Y_i\widehat{\textsc{G}}\left(X_i,\eta^{(1)},\theta_1\right)+o_p(1)\quad\text{and}\quad \widehat{S}_n^{(d)}=\frac{1}{n}\sum\limits_{k=1}^K\sum\limits_{i\in \mathbb{I}_k}\widehat{\textsc{G}}\left(X_i,\eta^{(1)},\theta_1\right)^2+o_p(1).
        \]
        Then, some algebra gives
        \begin{equation}
        \begin{aligned}
        \widehat{\sigma}^2_{\text{mdel}}
        =&\frac{1}{n}\sum\limits_{d=0,1}\sum\limits_{k=1}^K\sum\limits_{i\in \mathbb{I}^{(d)}_k}\frac{n_d}{n}\widehat{p}_i\left\{\frac{n}{n_1}D_i(Y_i-\eta^{(1)}(X_i)+\theta_1-\widehat{\theta}_{\text{mdel}}^{(1)})+\eta^{(1)}(X_i)-\theta_1-\right.\\
        &\left.\frac{n}{n_0}(1-D_i)(Y_i-\eta^{(0)}(X_i)+\theta_0-\widehat{\theta}_{\text{mdel}}^{(0)})-\eta^{(0)}(X_i)+\theta_0\right\}^2+o_p(1)\\
        =&\frac{1}{n^2}\sum\limits_{i=1}^n\left\{\frac{D_i}{\delta}\left(Y_i-\eta^{(1)}(X_i)\right)-\frac{1-D_i}{1-\delta}\left(Y_i-\eta^{(0)}(X_i)\right)+\eta^{(1)}(X_i)-\eta^{(0)}(X_i)-\theta\right\}^2+o_p(1).
        \end{aligned}
        \end{equation}
        This completes the proof.
        \end{proof}

\end{document}